\documentclass[authoryear,preprint,3p]{elsarticle}
\usepackage{amsmath,amssymb,indentfirst,epsfig,multirow,amsthm,color,array,amsfonts,graphicx,fullpage,hyperref,bm}
\usepackage[authoryear]{natbib}
\usepackage{amssymb,amsthm,epsfig, upgreek}
\usepackage{bm}

\usepackage[latin1]{inputenc}
\usepackage[figuresright]{rotating}
\usepackage{url,booktabs}
\newtheorem{TEO}{TEO}

\newtheorem{proposition}[TEO]{Proposition}
\usepackage{subfigure}

\journal{Computational Statistics \& Data Analysis}
\begin{document}

\begin{frontmatter}


\title{The compound class of extended Weibull power series distributions}


\author[label4]{Rodrigo B. Silva\fnref{label1}}
 \address[label4]{Universidade Federal de Pernambuco\\
 Departamento de Estatística, Cidade Universit\'aria, 50740-540 Recife, PE, Brazil
 \vspace{.1cm}}

 \author[label4]{Marcelo B. Pereira\fnref{label2}}

\author[label4]{Cícero R. B. Dias\fnref{label3} }

 \author[label4]{Gauss M. Cordeiro\fnref{label5}}

\fntext[label1]{Email: \texttt{rodrigobs29@gmail.com}}
\fntext[label2]{Email: \texttt{m.p.bourguignon@gmail.com}}
\fntext[label3]{Email: \texttt{cicerorafael@gmail.com}}
\fntext[label5]{Corresponding author. Email: \texttt{gauss@de.ufpe.br}}

\begin{abstract}
In this paper, we introduce a new class of distributions which is obtained by compounding the extended Weibull and power series distributions. The compounding procedure follows the same set-up carried out by Adamidis and Loukas (1998) and defines at least new 68 sub-models. This class includes some well-known mixing distributions, such as the Weibull power series  (Morais and Barreto-Souza, 2010) and exponential power series (Chahkandi and Ganjali, 2009) distributions. Some mathematical properties of the new class are studied including moments and generating function. We provide the density function of the order statistics and obtain their moments. The method of maximum likelihood is used for estimating the model parameters and an EM algorithm is proposed for computing the estimates. Special distributions are investigated in some detail. An application to a real data set is given to show the flexibility and potentiality of the new class of distributions.
\end{abstract}
\begin{keyword}
EM algorithm \sep Extended Weibull distribution \sep Extended Weibull power series distribution \sep Order statistic \sep Power series distribution.
\end{keyword}

\end{frontmatter}

\section{Introduction}

The modeling and analysis of lifetimes is an important aspect of statistical work in a wide variety of scientific and technological fields. Several distributions have been proposed in the literature
to model lifetime data by compounding some useful lifetime distributions. Adamidis and Loukas (1998) introduced a two-parameter exponential-geometric (EG) distribution by compounding an exponential distribution with a geometric distribution. In the same way, the exponential Poisson (EP) and exponential logarithmic (EL) distributions were introduced and studied by Kus (2007) and Tahmasbi and Rezaei (2008), respectively. Recently, Chahkandi and Ganjali (2009) proposed the exponential power series (EPS) fa\-mi\-ly of distributions, which contains as special cases these distributions. Barreto-Souza {\it et al.} (2010) and Lu and Shi (2011) introduced the Weibull-geometric (WG) and Weibull-Poisson (WP) distributions which naturally extend the EG and EP distributions, respectively.  In a very recent paper, Morais and Barreto-Souza (2011) defined the Weibull power series (WPS) class of distributions which contains the EPS distributions as sub-models. The WPS distributions can have an increasing, decreasing and upside down bathtub failure rate function.

Now, consider the class of extended Weibull (EW) distributions, as proposed by Gurvich {\it et al.} (1997), having the cumulative distribution function (cdf)
\begin{equation}\label{extweibull}
G(x;\, \alpha, \boldsymbol{\xi}) = 1 - \mathrm{e}^{-\alpha\, H(x;\, \boldsymbol{\xi})}, \quad x>0, \,\,\, \alpha>0,
\end{equation}
where $H(x;\, \boldsymbol{\xi})$ is a non-negative monotonically increasing function which depends on a pa\-ra\-me\-ter vector $\boldsymbol{\xi}$.  The corresponding probability density function (pdf) is given by
\begin{equation}\label{pdfextweibull}
g(x;\, \alpha, \boldsymbol{\xi}) = \alpha\, h(x;\, \boldsymbol{\xi}) \,\mathrm{e}^{-\alpha \,H(x;\,\boldsymbol{\xi})}, \quad x > 0, \,\,\, \alpha>0,
\end{equation}
where $h(x;\, \boldsymbol{\xi})$ is the derivative of $H(x; \,\boldsymbol{\xi})$.

Note that many well-known models are special cases of equation (\ref{extweibull}) such as:\\

\noindent ({\it i}) $H(x; \boldsymbol{\xi}) = x$ gives the exponential distribution;\\

\noindent ({\it ii}) $H(x; \boldsymbol{\xi}) = x^2$ yields the Rayleigh distribution (Burr type-X distribution);\\

\noindent ({\it iii}) $H(x; \boldsymbol{\xi}) = \log(x/k)$ leads to the Pareto distribution; \\

\noindent ({\it iv}) $H(x; \boldsymbol{\xi}) = \beta^{-1}[\exp(\beta x)-1]$ gives the Gompertz distribution.\\

In this article, we define the extended Weibull power series (EWPS) class of univariate distributions obtained by compounding the extended Weibull and power series distributions. The compounding procedure follows the key idea of Adamidis and Loukas (1998) or, more generally, by Chahkandi and Ganjali (2009) and Morais and Barreto-Souza {\it et al.} (2011). The new class of distributions contains as special models the WPS distributions, which in turn extends the EPS distributions and defines at least new 68 (17 $\times$ 4) sub-models as special cases. The hazard function of our class can be decreasing, increasing, bathtub and upside down bathtub.

We are motivated to introduce the EWPS distributions because of the wide usage of the general class of Weibull distributions and the fact that the current generalization provides means of its continuous extension to still more complex situations.

This paper is organized as follows. In Section 2, we define the EWPS class of distributions and demonstrate that there are many existing models which can be deduced as special cases of the proposed unified model. In Section 3, we provide the density, survival and hazard rate functions and derive some useful expansions. In Section 4, we obtain its quantiles, ordinary and incomplete moments. Further, the order statistics are discussed and their moments are determined. Section 5 deals with reliability and average lifetime. Estimation of the parameters by maximum likelihood using an EM algorithm and large sample inference are investigated in Section 6. In Section 7, we present suitable constraints leading to the maximum entropy characterization of the new class. Three special cases of the proposed class are studied in Section 8. In Section 9, we provide an application to a real data set. The paper is concluded in Section 10.

\section{The new class}

Our class can be derived as follows. Given $N$, let $X_1, \ldots, X_N$ be independent and identically distributed (iid) random variables following (\ref{extweibull}). Here, $N$ is a discrete random variable following a power series distribution (truncated at zero) with probability mass function
\begin{equation}\label{powerseries}
p_n = P(N=n)=\frac{a_n \, \theta^n}{C(\theta)}, n=1,2,\ldots,\\
\end{equation}
where $a_n$ depends only on $n$, $C(\theta) = \sum_{n=1}^{\infty}a_n \,\theta^n$ and $\theta>0$ is such that $C(\theta)$ is finite. Table \ref{table1} summarizes some power series distributions (truncated at zero) defined according to (\ref{powerseries}) such as the Poisson, logarithmic, geometric and binomial distributions. Let $X_{(1)} = \mbox{min}\left\{X_i\right\}^{N}_{i=1}$. The conditional cumulative distribution of $X_{(1)}|N = n$ is given by
\begin{equation*}
G_{X_{(1)}|N=n}(x) = 1-\mathrm{e}^{-n\alpha H(x; \boldsymbol{\xi})},
\end{equation*}
i.e.,  $X_{(1)}|N = n$ follows a general class of distributions~\eqref{extweibull} with parameters $n\alpha$ and $\boldsymbol{\xi}$ based on the same $H(x; \boldsymbol{\xi})$ function. Hence, we obtain
\begin{equation*}
P(X_{(1)} \leq x, N=n) = \frac{a_n\, \theta^n}{C(\theta)}\left[1-\mathrm{e}^{-n\alpha H(x; \boldsymbol{\xi})}\right], \quad x>0,
\quad n \geq 1.
\end{equation*}

The EWPS class of distributions can then be defined by the marginal cdf of $X_{(1)}$:

\begin{equation}\label{cdf}
F(x;\theta,\alpha, \boldsymbol{\xi}) = 1-\frac{C(\theta \,
\mathrm{e}^{-\alpha H(x; \boldsymbol{\xi})})}{C(\theta)}, \quad x>0.
\end{equation}

\begin{table}[!htbp]
\centering
\footnotesize
\begin{tabular}{llllllllllll}
 \toprule
Distribution && $a_n$&& $C(\theta)$ &&
$C'(\theta)$&&$C''(\theta)$&&$C(\theta)^{-1}$&$\Theta$\\
 \midrule
Poisson&&$n!^{-1}$ && $\mathrm{e}^{\theta}-1$&&$e^{\theta}$&&$e^{\theta}$&&$\log(\theta+1)$&$\theta \in (0, \infty)$\\
Logarithmic&&$n^{-1}$&&$-\log(1-\theta)$&&$(1-\theta)^{-1}$&&$(1-\theta)^{-2}$&&$1-\mathrm{e}^{-\theta}$&$\theta \in (0,1)$\\
Geometric&&1&&$\theta(1-\theta)^{-1}$&&$(1-\theta)^{-2}$&&$2(1-\theta)^{-3}$&&$\theta(\theta+1)^{-1}$&$\theta \in (0,1)$\\
Binomial&&$\binom{m}{n}$&&$(\theta+1)^{m} -
1$&&$m(\theta+1)^{m-1}$&&$\frac{m(m-1)}{(\theta+1)^{2-m}}$&&$(\theta-1)^{1/m}-1$&$\theta \in (0, 1)$\\
\bottomrule
\end{tabular}
\caption{Useful quantities for some power series distributions.}\label{table1}
\end{table}

\begin{sidewaystable}[!htbp]
\centering \footnotesize
\begin{tabular}{l@{            }l@{}llll} \toprule
Distribution & $H(x;\boldsymbol{\xi})$ & $h(x;\boldsymbol{\xi})$ & $\alpha$&$\boldsymbol{\xi}$&References \\
\midrule
Exponential ($x\geq0$)&$x$&1&$\alpha$&$\emptyset$&Johnson \emph{et al}. (1994) \\
Pareto ($x\geq k$)&$\log(x/k)$&$1/x$&$\alpha$&$k$&Johnson \emph{et al}. (1994)\\
Rayleigh ($x\geq0$) &$x^2$&$2x$&$\alpha$&$\emptyset$&Rayleigh (1880) \\
Weibull ($x\geq0$)&$x^{\gamma}$&$\gamma x^{\gamma-1}$&$\alpha$&$\gamma$&Johnson \emph{et al}. (1994) \\
Modified Weibull  ($x\geq0$) &$x^{\gamma}\exp(\lambda x)$ &
$x^{\gamma - 1}\exp(\lambda
x)(\gamma+\lambda x)$&$\alpha$&$[\gamma,\, \lambda]$&Lai \emph{et al}. (2003) \\
Weibull extension ($x\geq0$)&
$\lambda[\exp(x/\lambda)^{\beta}-1]$&$\beta\exp(x/\lambda)^{\beta}(x/\lambda)^{\beta-1}$&$\alpha$&$[\gamma,\, \lambda,\,\beta]$&Xie \emph{et al}. (2002) \\
Log-Weibull ($-\infty<x<\infty$)&$\exp[(x-\mu)/\sigma]$&$(1/\sigma)\exp[(x-\mu)/\sigma]$&1&$[\mu,\, \sigma]$&White (1969)  \\
Phani ($0<\mu < x<\sigma<\infty$)& $[(x-\mu)/(\sigma-x)]^{\beta}$&$\beta[(x-\mu)/(\sigma-x)]^{\beta-1}[(\sigma-\mu)/(\sigma-t)^2]$&$\alpha$&$[\mu,\, \sigma,\,\beta]$&Phani (1987)\\
Weibull Kies ($0<\mu < x<\sigma<\infty$)&$(x-\mu)^{\beta_1}/(\sigma-x)^{\beta_2}$&$(x-\mu)^{\beta_1 -1}(\sigma-x)^{-\beta_2 - 1}[\beta_1(\sigma-x)+\beta_2(x-\mu)]$&$\alpha$&$[\mu,\,\sigma,\,\beta_1,\,\beta_2]$&Kies (1958)\\
Additive Weibull ($x\geq0$)&$(x/\beta_1)^{\alpha_1} + (x/\beta_2)^{\alpha_2}$&$(\alpha_1/\beta_1)(x/\beta_1)^{\alpha_1 - 1} + (\alpha_2/\beta_2)(x/\beta_2)^{\alpha_2 - 1}$&1&$[\alpha_1,\,\alpha_2,\,\beta_1,\,\beta_2]$&Xie and Lai (1995)\\
Traditional Weibull ($x\geq0$)&$x^{b}[\exp(cx^{d}-1)]$&$b x^{b - 1}[\exp(cx^{d})-1] + cdx^{b + d -1}\exp(cx^d)$&$\alpha$&$[b,\,
c,\,d]$&Nadarajah and
Kotz (2005) \\
Gen. power Weibull ($x\geq0$)&$[1+(x/\beta)^{\alpha_1}]^{\theta} - 1$  &$(\theta\alpha_1/\beta)[1+(x/\beta)^{\alpha_1}]^{\theta-1}(x/\beta)^{\alpha_1}$&1&$[\alpha_1,\, \beta,\,\theta]$&Nikulin and Haghighi (2006) \\
Flexible Weibull extension($x\geq0$)&$\exp(\alpha_1x-\beta/x)$&$\exp(\alpha_1x-\beta/x)(\alpha_1+\beta/x^2)
$&1&$[\alpha_1,\,\beta]$&Bebbington \emph{et al}. (2007) \\
Gompertz ($x\geq0$)&$\beta^{-1}[\exp(\beta x) - 1]$&$\exp(\beta x)$&$\alpha$&$\beta$&Gompertz (1825) \\
Exponential power ($x\geq0$)&$\exp[(\lambda x)^{\beta}]-1$&$\beta\lambda\exp[(\lambda x)^{\beta}](\lambda x)^{\beta-1}$&1&$[\lambda,\,\beta]$&Smith and Bain (1975)\\
Chen ($x\geq0$)&$\exp(x^{b})-1$&$b x^{b-1}\exp(x^b)$&$\alpha$&$b$&Chen (2000) \\
Pham ($x\geq0$)&$ (a^{x})^{\beta} - 1$&$\beta(a^{x})^{\beta}\log(a)$&1&$[a,\,\beta]$&Pham (2002)\\
\bottomrule
\end{tabular}
\caption{Special distributions and related $H(x;\,\boldsymbol{\xi})$ and $h(x;\,\boldsymbol{\xi})$ functions.}
\label{table2}
\end{sidewaystable}

The random variable $X$ following (\ref{cdf}) with parameters $\theta$ and $\alpha$ and the vector $\boldsymbol{\xi}$ of parameters is denoted by $X\sim \mbox{EWPS}(\theta,\alpha, \boldsymbol{\xi})$. Equation (\ref{cdf}) extends several distributions which have been studied in the literature. The EG distribution (Adamidis and Loukas, 1998) is obtained by taking $H(x;\,\boldsymbol{\xi})=x$ and $C(\theta) = \theta(1-\theta)^{-1}$ with $\theta \in (0,1)$. Further, for $H(x;\,\boldsymbol{\xi})=x$, we obtain the EP (Kus, 2007) and EL  (Tahmasbi and Rezaei, 2008) distributions by taking $C(\theta) = \mathrm{e}^{\theta}-1, \theta >0$, and $C(\theta) = -\log(1-\theta), \theta \in (0,1)$, respectively. In the same way,
for $H(x;\,\boldsymbol{\xi})=x^\gamma$, we obtain the WG (Barreto-Souza {\it et al}., 2009) and WP (Lu and Shi, 2011) distributions. The EPS distributions are obtained from (\ref{cdf})
by mixing $H(x;\,\boldsymbol{\xi})=x$ with any $C(\theta)$ listed in Table \ref{table1} (see Chahkandi and Ganjali, 2009). Finally, we obtain the WPS distributions from (\ref{cdf}) by compounding $H(x;\,\boldsymbol{\xi})=x^\gamma$ with any $C(\theta)$ in Table \ref{table1} (see Morais and Barreto-Souza, 2011). Table \ref{table2} displays some useful quantities and respective parameter vectors for each particular distribution.


\section{Density, survival and hazard functions}

The density function associated to (\ref{cdf}) is given by
\begin{equation}\label{pdf}
f(x;\theta,\alpha, \boldsymbol{\xi}) = \theta \, \alpha \,
h(x;\boldsymbol{\xi})\, \mathrm{e}^{-\alpha H(x;\,
\boldsymbol{\xi})}\, \frac{C'(\theta \, \mathrm{e}^{-\alpha H(x;\,
\boldsymbol{\xi})})}{C(\theta)}, \quad x>0.
\end{equation}

\begin{proposition}\label{prop1}
The EW class of distributions with parameters $c\alpha$ and $\boldsymbol{\xi}$ is a limiting special case of the EWPS class of distributions when $\theta \rightarrow 0^+$, where $c= \min\left\{n \in \mathbb{N}: a_n >0\right\}$.
\end{proposition}
\begin{proof}
This proof uses a similar argument to that found in Morais and Barreto-Souza (2011). Define $c= \min\left\{n \in \mathbb{N}: a_n >0\right\}$. We have
\begin{align*}
\lim_{\theta \rightarrow 0^+} F(x) &= 1 - \lim_{\theta \rightarrow 0^+} \frac{\displaystyle{\sum_{n=c}^\infty a_n\left(\theta \, \mathrm{e}^{-\alpha H(x; \boldsymbol{\xi})}\right)^n}}{\displaystyle{\sum_{n=c}^\infty a_n\, \theta^n}}\\
&= 1 - \lim_{\theta \rightarrow 0^+} \frac{\displaystyle{\mathrm{e}^{-c \alpha H(x; \boldsymbol{\xi})} + a_c^{-1} \sum_{n=c+1}^\infty a_n \,\theta^{n-c}\mathrm{e}^{-n\alpha H(x; \boldsymbol{\xi})}}}{\displaystyle{1 + a_c^{-1}\sum_{n=c+1}^\infty a_n\, \theta^{n-c}}}\\
&= 1-\mathrm{e}^{-c\alpha H(x; \boldsymbol{\xi})},
\end{align*}
for $x>0$.
\end{proof}

We now provide an interesting expansion for ($\ref{pdf}$). We have $C'(\theta) = \sum_{n=1}^{\infty}n\, a_n\, \theta^{n-1}$. By using this result in (\ref{pdf}), it follows that
\begin{equation}\label{linearcomb}
f(x;\theta,\alpha, \boldsymbol{\xi}) = \sum_{n=1}^{\infty}p_n\,
g(x;\, n\alpha, \xi),
\end{equation}
where $g(x;\, n\alpha, \boldsymbol{\xi})$  is given by (\ref{pdfextweibull}). Based on equation~\eqref{linearcomb}, we obtain
\begin{equation*}
F(x;\theta,\alpha, \boldsymbol{\xi}) = 1 - \sum_{n=1}^\infty p_n\,
\mathrm{e}^{-n\alpha H(x;\, \boldsymbol{\xi})}.
\end{equation*}

Hence, the EWPS density function is an infinite mixture of EW densities. So, some mathematical quantities (such as ordinary and incomplete moments, generating function and mean deviations) of the EWPS distributions can be obtained by knowing those quantities for the baseline density function $g(x;\, n\alpha, \boldsymbol{\xi})$.

The EWPS survival function is given by
\begin{equation}\label{survewps}
S(x; \theta, \alpha, \boldsymbol{\xi})= \frac{C(\theta \,\mathrm{e}^{-\alpha H(x;\,\boldsymbol{\xi})})}{C(\theta)}
\end{equation}
and the corresponding hazard rate function becomes
\begin{equation}
\tau(x; \theta, \alpha, \boldsymbol{\xi}) = \theta \alpha\, h(x;\,
\boldsymbol{\xi}) \, \mathrm{e}^{-n\alpha H(x; \boldsymbol{\xi})}
\,\frac{C'(\theta \, \mathrm{e}^{-\alpha H(x;\,
\boldsymbol{\xi})})}{C(\theta \, \mathrm{e}^{-\alpha H(x;\,
\boldsymbol{\xi})})}.
\end{equation}

\section{Quantiles, moments and order statistics}

The EWPS distributions are easily simulated from (\ref{cdf}) as follows: if $U$ has a uniform $U(0,1)$ distribution, then the solution of the nonlinear equation
\begin{eqnarray*}
X = H^{-1}\left\{-\frac{1}{\alpha}\log\left[\frac{C^{-1}(C(\theta)(1-U))}{\theta}\right]\right\}
\end{eqnarray*}
has the EWPS$(\theta,\alpha, \boldsymbol{\xi})$ distribution, where
$H^{-1}(\cdot)$ and $C^{-1}(\cdot)$ are the inverse functions of
$H(\cdot)$ and $C(\cdot)$, respectively. To simulate data from this
nonlinear equation, we can use the matrix programming language Ox
through {\it SolveNLE} subroutine (see Doornik, 2007).

We now derive a general expression for the $r$th raw moment of $X$, which may be determined by using (\ref{linearcomb}) and the monotone convergence theorem. So, for $r \in \mathbb{N}$, we obtain
\begin{equation}
\operatorname{E}(X^r) = \sum_{n=1}^{\infty}p_n\, \operatorname{E}(Z^r),\\
\end{equation}
where $Z$ is a random variable with pdf $g(z; n\alpha,
\boldsymbol{\xi})$.

The incomplete moments and moment generating function (mgf) follow by using (\ref{linearcomb}) and the monotone convergence theorem:
\begin{align*}
I_X(y) &= \int^{y}_{0}x^r\, f(x)dx  = \sum_{n=1}^{\infty}p_n\, I_{Z}(y)\\
\intertext{and} M_{X}(t) &=
\sum_{n=1}^{\infty}p_n\,\operatorname{E}\left(\mathrm{e}^{tZ}\right).
\end{align*}
where $Z$ is defined as before.

Order statistics are among the most fundamental tools in non-parametric statistics and inference. They enter in the problems of estimation and hypothesis tests in a variety of ways. Therefore, we now discuss some properties of the order statistics for the proposed class of distributions. The pdf $f_{i:m}(x)$ of the $i$th order statistic for a random sample $X_1, \ldots, X_m$ from the EWPS distribution is given by
\begin{equation}\label{osewps}
f_{i:m}(x) = \frac{m!}{(i-1)!(m-i)!}f(x;\theta,\alpha, \boldsymbol{\xi})\left[1-\frac{C(\theta\, \mathrm{e}^{-\alpha H(x;\,
\boldsymbol{\xi})})}{C(\theta)}\right]^{i-1}\left[\frac{C(\theta
\mathrm{e}^{-\alpha H(x;\,
\boldsymbol{\xi})})}{C(\theta)}\right]^{m-i}, \quad x> 0,
\end{equation}
where $f (x; \theta, \alpha, \boldsymbol{\xi})$ is the pdf given by (\ref{pdf}). By using the binomial expansion, we can write (\ref{osewps}) as
\begin{equation}\label{osewpsexpansion}
f_{i:m}(x) = \frac{m!}{(i-1)!(m-i)!}f(x;\theta,\alpha, \boldsymbol{\xi})\sum_{j=0}^{i-1} (-1)^j \,\binom{i-1}{j}\, S(x;\theta,\alpha,
\boldsymbol{\xi})^{m+j-i},
\end{equation}
where $S(x; \theta, \alpha, \boldsymbol{\xi})$ is given by~\eqref{survewps}. The corresponding cumulative function is
\begin{equation*}
F_{i:m}(x) = \sum_{j=0}^{\infty}\sum_{k=i}^{m}(-1)^j\,\binom{k}{j}\,\binom{m}{k}\,S(x;\theta,\alpha,
\boldsymbol{\xi})^{m+j-k}.
\end{equation*}

An alternative form for~\eqref{osewps} can be obtained from~\eqref{linearcomb} as
\begin{equation}\label{altosewps}
f_{i:m}(x) = \frac{m!}{(i-1)! (m-i)!} \sum_{n=1}^\infty \sum_{j=0}^{i-1} \omega_j\, p_n\, g(x; n\alpha,
\boldsymbol{\xi})S(x;\theta,\alpha, \boldsymbol{\xi})^{m+j-1},
\end{equation}
where $\omega_j = (-1)^j \binom{i-1}{j}$. So, the $s$th raw moment $X_{i:m}$ comes immediately from the above equation
\begin{equation}\label{eosewps}
\operatorname{E}\left(X_{i:m}^s\right) = \frac{m!}{(i-1)! (m-i)!}
\sum_{n=1}^\infty \sum_{j=0}^{i-1} \omega_j\, p_n\,
\operatorname{E}\left[Z^s S(Z)^{m+j-i}\right],
\end{equation}
where $Z\sim \mbox{EW}(n\alpha, \boldsymbol{\xi})$ is defined before.
\section{Reliability and average lifetime}

In the context of reliability, the stress-strength model describes the life of a component which has a random strength $X$ subjected to a random stress $Y$. The component fails at the instant that the stress applied to it exceeds the strength, and the component will function satisfactorily whenever $X > Y$. Hence, $R = \operatorname{P}(X > Y)$ is a measure of component reliability. It has many applications, especially in engineering concepts. The algebraic form for R has been worked out for the majority of the well-known distributions. Here, we obtain the form for the reliability $R$ when $X$ and $Y$ are independent random variables having the same EWPS distribution.

The quantity $R$ can be expressed as
\begin{equation}\label{eqR}
R = \int_0^{\infty}f(x; \theta, \alpha, \boldsymbol{\xi})F(x;\theta,\alpha, \boldsymbol{\xi}) dx.
\end{equation}

Substituting (\ref{cdf}) and (\ref{pdf}) into equation (\ref{eqR}), we obtain
\begin{eqnarray*}\label{eqR2}
R &=& \int_{0}^{\infty}\theta \, \alpha \,
h(x;\boldsymbol{\xi})\, \mathrm{e}^{-\alpha H(x;\,
\boldsymbol{\xi})}\, \frac{C'(\theta \mathrm{e}^{-\alpha H(x;\,
\boldsymbol{\xi})})}{C(\theta)}\left[1-\frac{C(\theta
\mathrm{e}^{-\alpha H(x; \boldsymbol{\xi})})}{C(\theta)}\right] dx\\
&=&1 - \sum_{n=1}^{\infty}p_n\int_{0}^{\infty}
g(x;n\alpha,\boldsymbol{\xi})S(x;\theta,\alpha, \boldsymbol{\xi})dx,
\end{eqnarray*}
where the integral can be calculated from the baseline EW distribution.

The average lifetime is given by
\begin{equation*}
t_m =
\sum_{n=1}^{\infty}p_n\int\limits_{0}^{\infty}\operatorname{e}^{-n\alpha
H(x;\, \boldsymbol{\xi})}dx.
\end{equation*}

Given that there was no failure prior to $x_0$, the residual life is the period from time $x_0$ until the time of failure. The mean residual lifetime can be expressed as
\begin{eqnarray*}
m(x_0;\theta,\alpha,\boldsymbol{\xi}) &=& \left[\operatorname{Pr}(X>x_0)\right]^{-1}\int\limits_{0}^{\infty}y\,f(x_0+y;\theta,\alpha,\boldsymbol{\xi}) dy\\
&=&[S(x_0)]^{-1}\sum_{n=1}^{\infty}p_n\int\limits_{0}^{\infty}y\,g(x_0+y;n\alpha,\boldsymbol{\xi}) dy.
\end{eqnarray*}

The last integral can be computed from the baseline EW distribution. Furthermore, $m(x_0;\theta,\alpha,\boldsymbol{\xi}) \rightarrow \operatorname{E}(X)$ as $x_0 \rightarrow 0$.


\section{Maximum likelihood estimation}

\subsection{Preliminaries}
\vskip3mm Here, we determine the maximum likelihood estimates (MLEs) of the parameters of the EWPS class of distributions from complete samples only. Let $X_1, \ldots, X_n$ be a random sample with observed values $x_1, \ldots, x_n$ from an EWPS distribution with parameters $\theta, \alpha$ and $\boldsymbol{\xi}$. Let $\Theta= (\theta,\alpha, \boldsymbol{\xi})^\top$ be the $p \times 1$ parameter vector. The total log-likelihood function is given by
\begin{eqnarray}\label{loglik} \nonumber
\ell_n &=& \ell_{n}(x; \Theta) = n\left[\log\theta + \log\alpha - \log C(\theta)\right] - \alpha\sum_{i=1}^{n}H(x_i;\, \boldsymbol{\xi}) +
\sum_{i=1}^{n}\log h(x_i;\, \boldsymbol{\xi})\\
 &+& \sum_{i=1}^{n}\log C'(\theta\, \mathrm{e}^{-\alpha H(x_i;\, \boldsymbol{\xi})}).
\end{eqnarray}

The log-likelihood can be maximized either directly by using the SAS (PROC NLMIXED) or the Ox program (sub-routine MaxBFGS) (see Doornik, 2007) or by solving the nonlinear likelihood
equations obtained by differentiating (\ref{loglik}). The components of the score function $U_n(\Theta) = \left(\partial \ell_n/\partial \theta, \partial \ell_n/\partial \alpha, \partial \ell_n/\partial \xi\right)^\top$ are
\begin{align*}
\frac{\partial \ell_n}{\partial \alpha} &= \frac{n}{\alpha} - \sum_{i=1}^{n} H(x_i;\, \boldsymbol{\xi}) - \theta\sum_{i=1}^n H(x_i;\, \boldsymbol{\xi}) \mathrm{e}^{-\alpha H(x_i;\, \boldsymbol{\xi})}\,\frac{C''(\theta\, \mathrm{e}^{-\alpha H(x_i;\, \boldsymbol{\xi})})}{C'(\theta\, \mathrm{e}^{-\alpha H(x_i;\, \boldsymbol{\xi})})},\\
\frac{\partial \ell_n}{\partial \theta} &= \frac{n}{\theta} - n\frac{C'(\theta)}{C(\theta)} + \sum_{i=1}^n \mathrm{e}^{-\alpha H(x_i;\, \boldsymbol{\xi})}\,\frac{C''(\theta\, \mathrm{e}^{-\alpha H(x_i;\, \boldsymbol{\xi})})}{C'(\theta\, \mathrm{e}^{-\alpha H(x_i;\, \boldsymbol{\xi})})}\\
\intertext{and} \frac{\partial \ell_n}{\partial \boldsymbol{\xi}_k} &=
\sum_{i=1}^{n} \frac{\partial \log h(x_i;\, \boldsymbol{\xi})}{\partial \boldsymbol{\xi}_k} - \alpha \sum_{i=1}^{n} \frac{\partial H(x_i;\, \xi)}{\partial
\boldsymbol{\xi}_k}\left[1 + \theta \mathrm{e}^{-\alpha H(x_i; \, \boldsymbol{\xi})} \frac{C''(\theta\, \mathrm{e}^{-\alpha H(x_i;\,\boldsymbol{\xi})})}{C'(\theta\, \mathrm{e}^{-\alpha H(x_i;\,
\boldsymbol{\xi})})}\right].
\end{align*}

For interval estimation on the model parameters, we require the observed information matrix
\[ J_n(\Theta)=-\left( \begin{array}{cccc}
U_{\theta\theta} & U_{\theta\alpha} & |&U_{\theta\boldsymbol{\xi}}^\top \\
U_{\alpha\theta} & U_{\alpha\alpha} & |&U_{\alpha\boldsymbol{\xi}}^\top \\
-- & --& --& --\\
U_{\theta\boldsymbol{\xi}} & U_{\alpha\boldsymbol{\xi}} & | & U_{\boldsymbol{\xi}\boldsymbol{\xi}} \end{array} \right),\]
whose elements are listed in \ref{apA}. Let $\widehat{\Theta}$ be the MLE of $\Theta$. Under standard regular conditions stated in Cox and Hinkley (1974) that are fulfilled for our model whenever the parameters are in the interior of the parameter space, we have that the asymptotic distribution of $\sqrt{n}\left(\widehat{\Theta} - \Theta\right)$ is multivariate normal $N_p(0, K(\Theta)^{-1})$, where $K(\Theta) = \lim_{n \rightarrow \infty}J_n(\Theta)$ is the unit information matrix and $p$ is the number of parameters of the compounded distribution.

\subsection{The EM algorithm}
\vskip3mm
Here, we propose an EM algorithm (Dempster {\it et al.}, 1977) to estimate $\Theta$.
The EM algorithm is a recurrent method such that each step consists of an estimate
of the expected value of a hypothetical random va\-ria\-ble and then maximizes the log-likelihood for the complete data. Let the complete-data be $X_1, \ldots, X_n$ with observed values $x_1, \ldots, x_n$ and the hypothetical random va\-ria\-bles $Z_1, \ldots, Z_n$. The joint probability function is such that the marginal density of $X_1, \ldots, X_n$ is the likelihood of interest. Then, we define a hypothetical complete-data distribution for each $(X_i, Z_i)^\top, i=1, \ldots, n$, with a joint probability function in the form
\begin{equation*}
g(x, z; \Theta) = \frac{\alpha\, z\, a_z\, \theta^z}{C(\theta)}\,h(x;\, \xi)\, \mathrm{e}^{-\alpha z H(x;\, \xi)},
\end{equation*}
where $\theta$ and $\alpha$ are positive, $x>0$ and $z \in \mathbb{N}$. Under this formulation, the E-step of an EM cycle requires the expectation of $Z|X$;
$\Theta^{(r)} = (\theta^{(r)}, \alpha^{(r)},
\boldsymbol{\xi}^{(r)})^\top$ as the current estimate (in the rth
iteration) of $\Theta$. The probability function of $Z$ given $X$,
say $g(z|x)$, is given by
\begin{equation*}
g(z|x) = \frac{z\, a_z\, \theta^{\theta-1}}{C'(\theta e^{-\alpha H(x_i;\, \boldsymbol{\xi})})}\,\mathrm{e}^{-\alpha(z-1) H(x_i;\, \boldsymbol{\xi})}
\end{equation*}
and its expected value is
\begin{equation*}
E(Z|X) = 1 + \theta e^{-\alpha H(x;\, \boldsymbol{\xi})}\,\frac{C''(\theta\, \mathrm{e}^{-\alpha H(x;\, \boldsymbol{\xi})})}{C'(\theta\, \mathrm{e}^{-\alpha H(x;\, \boldsymbol{\xi})})}.
\end{equation*}

The EM cycle is completed with the M-step by using the maximum likelihood estimation over $\Theta$, where the missing $Z's$ are replaced by their conditional expectations given before. The log-likelihood for the complete-data is
\begin{align*}
\textstyle
\ell_n^*(x_1, \ldots, x_n; \, z_1, \ldots, z_n; \,\alpha, \,\theta, \,\boldsymbol{\xi}) &\propto n\log\alpha + \log \theta \sum_{i=1}^n z_i + \sum_{i=1}^n \log h(x_i;\, \boldsymbol{\xi})\\
 &- \alpha \sum_{i=1}^n z_i H(x_i;\, \boldsymbol{\xi}) - n\log C(\theta).
\end{align*}

So, the components of the score function $U^*_n(\Theta) = \left(\partial l^*_n/\partial \theta, \partial l^*_n/\partial \alpha, \partial l^*_n/\partial \boldsymbol{\xi}\right)^\top$ are
\begin{align*}
\frac{\partial l^*_n}{\partial \theta} &= \frac{n}{\theta} -
\sum_{i=1}^{n} z_i - n\frac{C'(\theta)}{C(\theta)}, \quad  \quad
\frac{\partial l^*_n}{\partial \alpha} = \frac{n}{\alpha} -
\sum_{i=1}^{n} z_i H(x_i;\, \boldsymbol{\xi}) \quad \intertext{and}
\frac{\partial l^*_n}{\partial \boldsymbol{\xi}_k} &= \sum_{i=1}^{n}
\frac{\partial \log h(x_i;\, \boldsymbol{\xi})}{\partial \xi_k} - \alpha \sum_{i=1}^{n} z_i \frac{\partial H(x_i;\,\boldsymbol{\xi})}{\partial \boldsymbol{\xi}_k}.
\end{align*}

From a nonlinear system of equations $U^*_n(\widehat{\Theta}) = 0$, we obtain the iterative procedure of the EM algorithm
\begin{align*}
&\hat{\alpha}^{(t+1)} = \frac{n}{\sum_{i=1}^{n} z_i^{(t)} H(x_i;\,\boldsymbol{\xi}^{(t)})}, \quad \quad \hat{\theta}^{(t+1)} =
\frac{C(\hat{\theta}^{(t+1)})}{C'(\hat{\theta}^{(t+1)})}\frac{1}{n}\sum_{i=1}^{n}
z_i^{(t)} \intertext{and} &\sum_{i=1}^{n} \frac{\partial \log h(x_i;\,\hat{\boldsymbol{\xi}}^{(t+1)})}{\partial \boldsymbol{\xi}_k} -
\hat{\alpha}^{(t)}\sum_{i=1}^{n} z_i^{(t)}\frac{\partial H(x_i;\,\hat{\boldsymbol{\xi}}^{(t+1)})}{\partial \boldsymbol{\xi}_k} = 0,
\end{align*}
where $\hat{\theta}^{(t+1)}, \hat{\alpha}^{(t+1)}$ and
$\hat{\xi}^{(t+1)}$ are obtained numerically. Here, for $i=1,
\ldots, n$, we have
\begin{equation*}
z_i^{(t)} = 1 +\hat{\theta}^{(t)} \mathrm{e}^{-\hat{\alpha}^{(t)} H(x_i;\,\hat{\boldsymbol{\xi}}^{(t)})} \frac{C''(\hat{\theta}^{(t)}
\mathrm{e}^{-\hat{\alpha}^{(t)} H(x_i;\,\hat{\boldsymbol{\xi}}^{(t)})})}{C'(\hat{\theta}^{(t)}
\mathrm{e}^{-\hat{\alpha}^{(t)} H(x_i;\,\hat{\boldsymbol{\xi}}^{(t)})})}.
\end{equation*}

Note that, in each step, $\theta, \alpha$ and $\boldsymbol{\xi}$ are estimated independently. The EWPS distributions can be very useful in modeling lifetime data and practitioners may be interested in fitting one of our models.

\section{Maximum entropy identification}

Shannon (1948) introduced the probabilistic definition of entropy which is closely connected with the definition of entropy in statistical mechanics. Let $X$ be a random variable of a continuous distribution with density $f$. Then, the Shannon entropy of $X$ is defined by
\begin{equation}\label{shannon}
\mathbb{H}_{Sh}(f) = - \int_{\mathbb{R}}f(x;\theta,\alpha,
\boldsymbol{\xi})\log\left[f(x;\theta,\alpha, \boldsymbol{\xi})\right] dx.
\end{equation}

Jaynes (1957) introduced one of the most powerful techniques employed in the field of probability and statistics called the maximum entropy method. This method is closely related to the Shannon entropy and considers a class of density functions
\begin{equation}\label{jaynes}
\mathbb{F} = \left\{f(x;\theta,\alpha, \boldsymbol{\xi}):
\operatorname{E}_{f}(T_i(X)) = \alpha_i,\, i = 0, \ldots, m\right\},
\end{equation}
where $T_i (X), i = 1, \ldots, m$, are absolutely integrable functions with respect to $f$, and $T_0(X) = a_0 = 1$. In the continuous case, the maximum entropy principle suggests deriving the unknown density function of the random variable $X$ by the model that maximizes the Shannon entropy in~\eqref{shannon}, subject to the information constraints defined in the class $\mathbb{F}$. Shore and Johnson (1980) treated axiomatically the maximum entropy method. This method has been successfully applied in a wide variety of fields and has also been used for the characterization of several standard probability distributions; see, for example, Kapur (1989), Soofi (2000) and Zografos and Balakrishnan (2009).

The maximum entropy distribution is the density of the class F, denoted by $f^{ME}$, which is obtained as the solution of the
optimization problem
\begin{equation*}
f^{ME}(x;\theta,\alpha, \boldsymbol{\xi}) = \arg \max_{f \in \mathbb{F}} \mathbb{H}_{Sh}.
\end{equation*}

Jaynes (1957, p. 623) states that the maximum entropy distribution $f^{ME}$, obtained by the constrained maximization problem described above, ``is the only unbiased assignment we can make; to use any other would amount to arbitrary assumption of information which by hypothesis we do not have". It is the distribution which should not incorporate
additional exterior information other than which is specified by the constraints.

We now derive suitable constraints in order to provide a maximum entropy characterization for our class of distributions defined by~\eqref{cdf}. For this purpose, the next result plays an important role.

\begin{proposition}\label{constraints}
Let X be a random variable with pdf given by~\eqref{pdf}. Then,
\begin{description}
\item[C1.] $\operatorname{E}\left[\log( C'(\theta \, \mathrm{e}^{-\alpha H(X;\, \boldsymbol{\xi})}))\right] = \dfrac{\theta}{C(\theta)}\operatorname{E}\left[C'(\theta \, \mathrm{e}^{-\alpha H(Y;\, \boldsymbol{\xi})})\log( C'(\theta \, \mathrm{e}^{-\alpha H(Y;\, \boldsymbol{\xi})}))\right];$

\item[C2.] $\operatorname{E}\left[\log( h(X; \,\boldsymbol{\xi}))\right] = \dfrac{\theta}{C(\theta)}\operatorname{E}\left[C'(\theta \, \mathrm{e}^{-\alpha H(Y;\, \boldsymbol{\xi})})\log( h(Y;\,\boldsymbol{\xi}) )\right];$

\item[C3.] $\operatorname{E}\left[H(X; \,\boldsymbol{\xi})\right] = \dfrac{\theta}{C(\theta)}\operatorname{E}\left[C'(\theta \, \mathrm{e}^{-\alpha H(Y;\, \boldsymbol{\xi})})H(Y;\,\boldsymbol{\xi} )\right],$
\end{description}
where Y follows the EW distribution with density function~\eqref{pdfextweibull}.
\end{proposition}
\begin{proof}
The constraints C1,  C2 and C3 are easily obtained and therefore their demonstrations are omitted.
\end{proof}

The next proposition reveals that the EWPS distribution has maximum entropy in the class of all probability distributions specified by the constraints stated in the previous proposition.
\begin{proposition}
The pdf f of a random variable X, given by~\eqref{pdf}, is the unique solution of the optimization problem
\begin{equation*}
f(x;\theta,\alpha, \boldsymbol{\xi}) = \arg \max_{h}
\mathbb{H}_{Sh},
\end{equation*}
under the constraints $\rm{C1}$, $\rm{C2}$ and $\rm{C3}$ presented in the Proposition~\ref{constraints}.
\end{proposition}
\begin{proof}
Let $\tau$ be a pdf which satisfies the constraints C1, C2 and C3. The Kullback-Leibler divergence between $\tau$ and $f$ is
\begin{equation*}
D(\tau, f) = \int_{\mathbb{R}} \tau(x;\theta,\alpha, \boldsymbol{\xi}) \log\left(\frac{\tau(x;\theta,\alpha, \boldsymbol{\xi})}{f(x;\theta,\alpha, \boldsymbol{\xi})}\right) dx.
\end{equation*}

Following Cover and Thomas (1991), we obtain
\begin{eqnarray*}
0 \leq D(\tau, f) &=& \int_{\mathbb{R}} \tau(x;\theta,\alpha, \boldsymbol{\xi}) \log\left[ \tau(x;\theta,\alpha, \boldsymbol{\xi})\right] dx - \int_{\mathbb{R}} \tau(x;\theta,\alpha, \boldsymbol{\xi}) \log\left[f(x;\theta,\alpha, \boldsymbol{\xi})\right] dx\\
&=& -\mathbb{H}_{Sh}(\tau;\theta,\alpha, \boldsymbol{\xi}) - \int_{\mathbb{R}} \tau(x;\theta,\alpha, \boldsymbol{\xi}) \log\left[
f(x;\theta,\alpha, \boldsymbol{\xi})\right]dx.
\end{eqnarray*}
From the definition of $f$ and based on the constraints C1, C2 and C3, it follows that
\begin{eqnarray*}
\hspace{-2cm}\int_{\mathbb{R}} \tau(x) \log\left[f(x)\right] dx &=& \log(\theta \alpha) + \frac{\theta}{C(\theta)}\operatorname{E}\left\{C'(\theta \, \mathrm{e}^{-\alpha H(Y;\, \boldsymbol{\xi})})\log\left[ h(Y;\,\boldsymbol{\xi}) \right]\right\} - \log\left[ C(\theta)\right]\\
&-&\alpha \frac{\theta}{C(\theta)}\operatorname{E}\left[C'(\theta \, \mathrm{e}^{-\alpha H(Y;\, \boldsymbol{\xi})})H(Y;\,\boldsymbol{\xi} )\right] \\
&+& \frac{\theta}{C(\theta)}\operatorname{E}\left\{\log\left[ C'(\theta \, \mathrm{e}^{-\alpha H(Y;\, \boldsymbol{\xi})})\right] C'(\theta \, \mathrm{e}^{-\alpha H(Y;\, \boldsymbol{\xi})})\right\}\\
&=&  \int_{\mathbb{R}} f(x;\theta,\alpha, \boldsymbol{\xi}) \log\left[
f(x;\theta,\alpha, \boldsymbol{\xi})\right] dx  = - \mathbb{H}_{Sh}(f),
\end{eqnarray*}
where $Y$ is defined as before. So, we have
$\mathbb{H}_{Sh}(\tau) \leq \mathbb{H}_{Sh}(f)$ with equality if and
only if $\tau(x;\theta,\alpha, \boldsymbol{\xi}) = f
(x;\theta,\alpha, \boldsymbol{\xi})$ for all $x$, except for a set of measure 0, thus proving the uniqueness.
\end{proof}

The intermediate steps in the above proof in fact provide the following explicit expression for the Shannon entropy of the EWPS distribution
\begin{eqnarray}\nonumber
&&\mathbb{H}_{Sh}(f) = -\log(\theta \alpha) - \frac{\theta}{C(\theta)}\operatorname{E}\left\{C'(\theta \, \mathrm{e}^{-\alpha H(Y;\, \boldsymbol{\xi})})\log\left[h(Y;\,\boldsymbol{\xi}) \right]\right\}+ \log\left[C(\theta)\right]\\
&&+\alpha \frac{\theta}{C(\theta)}\operatorname{E}\left[C'(\theta \, \mathrm{e}^{-\alpha H(Y;\, \boldsymbol{\xi})})H(Y;\,\boldsymbol{\xi} )\right] - \frac{\theta}{C(\theta)}\operatorname{E}\left\{C'(\theta \, \mathrm{e}^{-\alpha H(Y;\, \boldsymbol{\xi})})\log\left[C'(\theta \, \mathrm{e}^{-\alpha H(Y;\, \boldsymbol{\xi})})\right]
\right\}.
\end{eqnarray}

For some EWPS distributions, the above results can only be obtained numerically.
\section{Special models}

In this section, we investigate some special cases of the EWPS class of distributions. We offer some expressions for moments and moments of the order statistics. To illustrate the flexibility of these distributions, we provide plots of the density and hazard rate functions for selected parameter values.

\subsection{Modified Weibull geometric distribution}
\vskip3mm
The modified Weibull geometric (MWG) distribution is defined by the cdf (\ref{cdf}) with $H(x;\, \boldsymbol{\xi}) = x^\gamma$ and $C(\theta) = \theta(1-\theta)^{-1}$ leading to
\begin{equation}
F(x;\theta,\alpha, \gamma,\lambda) = 1 -
\frac{(1-\theta)\exp\left(-\alpha x^\gamma \mathrm{e}^{\lambda
x}\right)}{1 - \theta \exp\left(-\alpha x^\gamma \mathrm{e}^{\lambda
x}\right)}, \quad x > 0,
\end{equation}
where $\theta \in (0,1)$. The associated pdf and hazard rate function are
\begin{align*}
f(x;\theta,\alpha, \gamma,\lambda) &= \alpha (1-\theta) (\gamma +
\lambda x)\,x^{\gamma-1}\frac{\exp\left(\lambda x - \alpha x^\gamma
\mathrm{e}^{\lambda x}\right)}{\left[1-\theta \exp\left(-\alpha
x^\gamma \mathrm{e}^{\lambda x}\right)\right]^2} \intertext{and}
\tau(x;\theta,\alpha, \gamma,\lambda) &= \alpha(\gamma+\lambda
x)\,x^{\gamma-1}\frac{\exp\left(\lambda x\right)}{1 - \theta
\exp\left(-\alpha x^\gamma \mathrm{e}^{\lambda x}\right)}
\end{align*}
for $x > 0$, respectively. The MWG distribution contains the WG distribution (Barreto-Souza \emph{et al}. (2010)) as the particular choice $\lambda = 0$. Further, for $\lambda = 0$ and $\alpha=1$, we obtain the EG distribution (Adamidis and Loukas (1998)). Figures $\ref{fig:densityfigewps}$ and $\ref{fig:hazardfigewps}$ display the density and hazard functions of the MWG distribution for selected parameter values.

\begin{figure}[!htbp]
\centering
\begin{tabular}{lll}
\hspace{-0.7cm}\subfigure[\scriptsize{$\alpha= 2, \,\gamma = 1.5\,\,\mbox{and}\,\, \lambda=0.5$} \label{ds1}]{\includegraphics[width=0.37\textwidth]{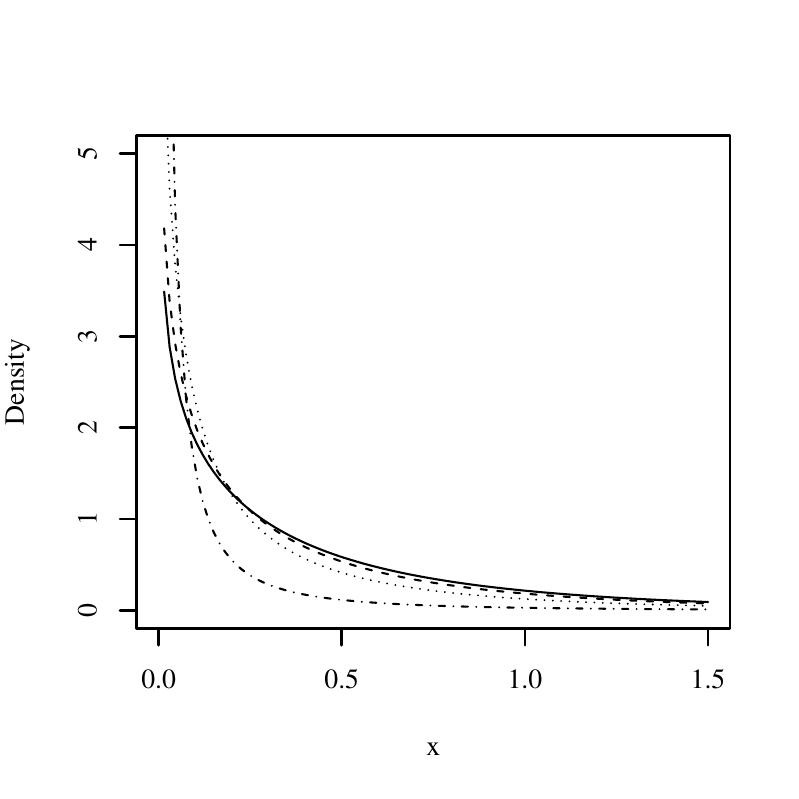}}&\hspace{-0.6cm}\subfigure[\scriptsize{$\alpha= 2, \, \gamma = 0.8\,\,\mbox{and}\,\, \lambda=0.01$} \label{ds2}]{\includegraphics[width=0.37\textwidth]{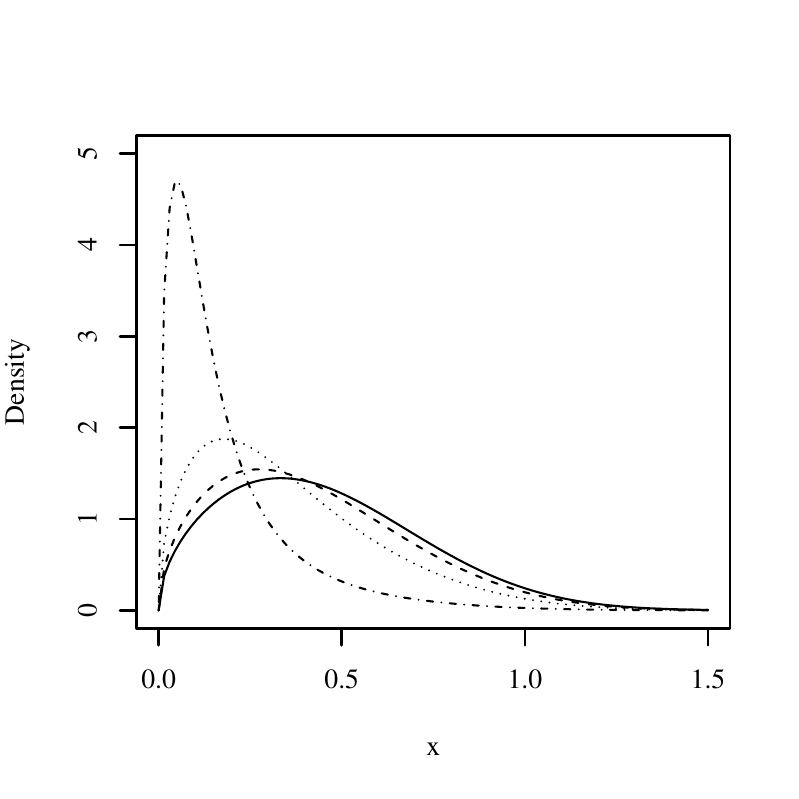}}&\hspace{-0.6cm}
\subfigure[\scriptsize{$\alpha= 0.1, \, \gamma = 6\,\,\mbox{and}\,\, \lambda=0.5$} \label{ds3}]{\includegraphics[width=0.37\textwidth]{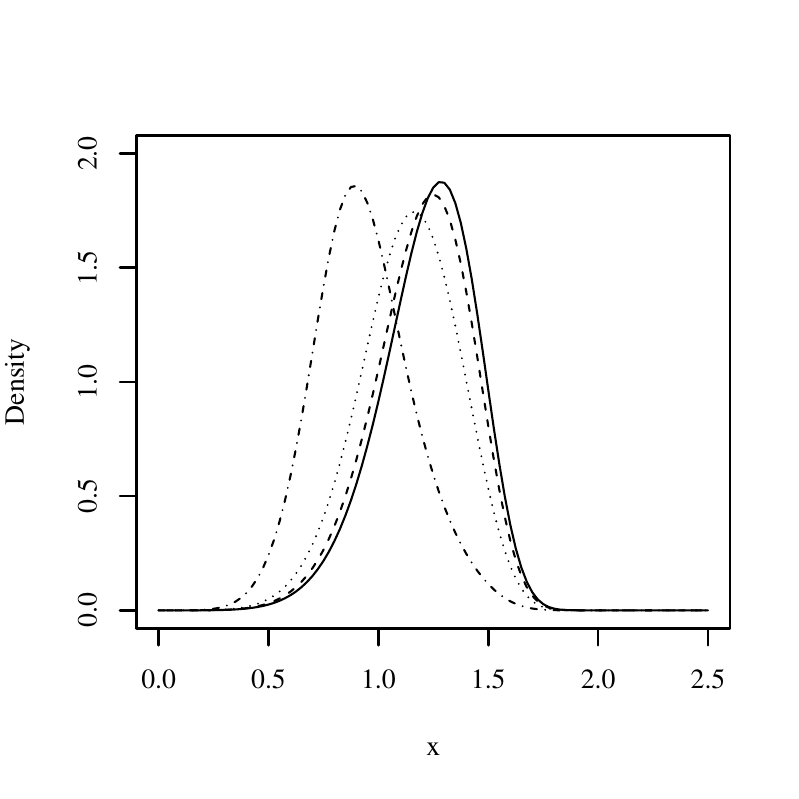}}\\
\end{tabular}
\vspace{-0.25cm}
\caption{Plots of the MWG density functions for $\theta = 0.01$ (solid line), $\theta = 0.2$ (dashed
line), $\theta = 0.5$ (dotted line) and $\theta = 0.9$ (dotdash line).}
\label{fig:densityfigewps}
\end{figure}

The $r$th raw moment of the random variable $X$ having the MWG distribution has closed-form. It is calculated from (\ref{linearcomb}) as
\begin{equation}\label{momentmwg}
E(X^r) = \sum_{n=1}^{\infty} p_n\, \mu_r(n),
\end{equation}
where $\mu_r(n) = \int_{0}^{\infty}x^r g(x;\, n\alpha, \gamma, \lambda)dx$ denotes the $r$th raw moment of the MW distribution with pa\-ra\-me\-ters $n\alpha, \gamma$ and $\lambda$. Here $p_n$ corresponds to the probability function of the geometric distribution. Carrasco {\it et al.} (2008) determined an infinite representation for the $r$th raw moment of the MW distribution with these parameters expressed as
\begin{equation}\label{carrascoetal}
\mu_r(n) = \sum_{i_1, \ldots, i_r=1}^{\infty} \frac{A_{i_1, \ldots, i_r}\,\Gamma(s_r/\gamma + 1)}{(n\alpha)^{s_r/\gamma}},
\end{equation}
where
\begin{equation*}
A_{i_1, \ldots, i_r} = a_{i_1}, \ldots, a_{i_r} \,\,\,\,\, \mbox{and} \,\,\,\,\, s_r = i_1, \ldots, i_r,
\end{equation*}
and
\begin{equation*}
a_i = \frac{(-1)^{i+1}i^{i-2}}{(i-1)!}\left(\frac{\lambda}{\gamma}\right)^{i-1}.
\end{equation*}
Hence, the moments of the MWG distribution can be obtained directly from equations (\ref{momentmwg}) and (\ref{carrascoetal}).
\begin{figure}[!htbp]
\centering
\begin{tabular}{lll}
\hspace{-0.7cm}\subfigure[\scriptsize{$\alpha= 2, \,\gamma = 1.5\,\,\mbox{and}\,\, \lambda=0.5$} \label{hz1}]{\includegraphics[width=0.37\textwidth]{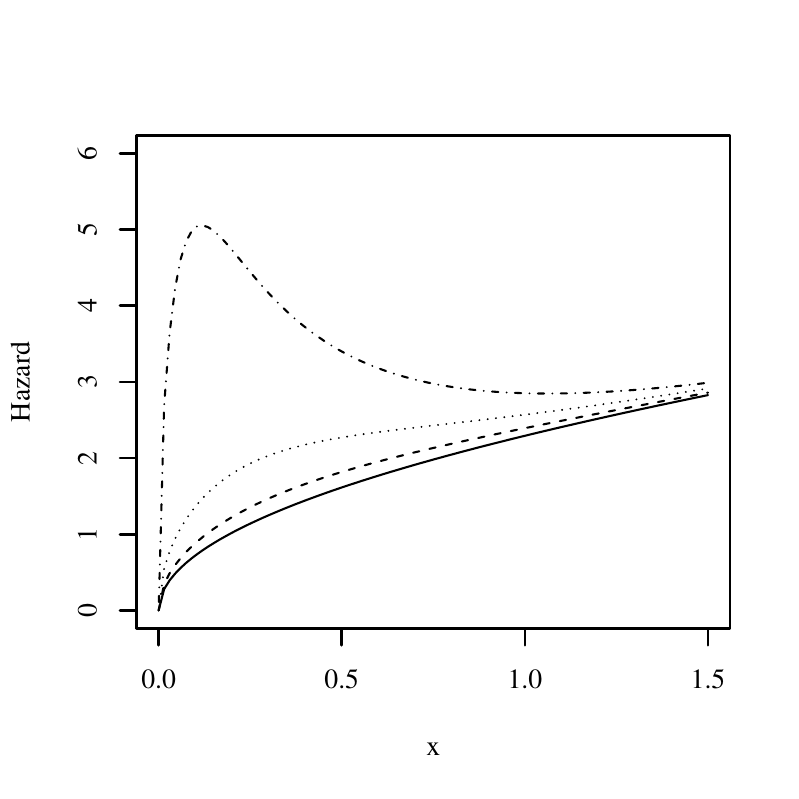}}&\hspace{-0.6cm}\subfigure[\scriptsize{$\alpha= 2, \, \gamma = 0.8\,\,\mbox{and}\,\, \lambda=0.01$} \label{hz2}]{\includegraphics[width=0.37\textwidth]{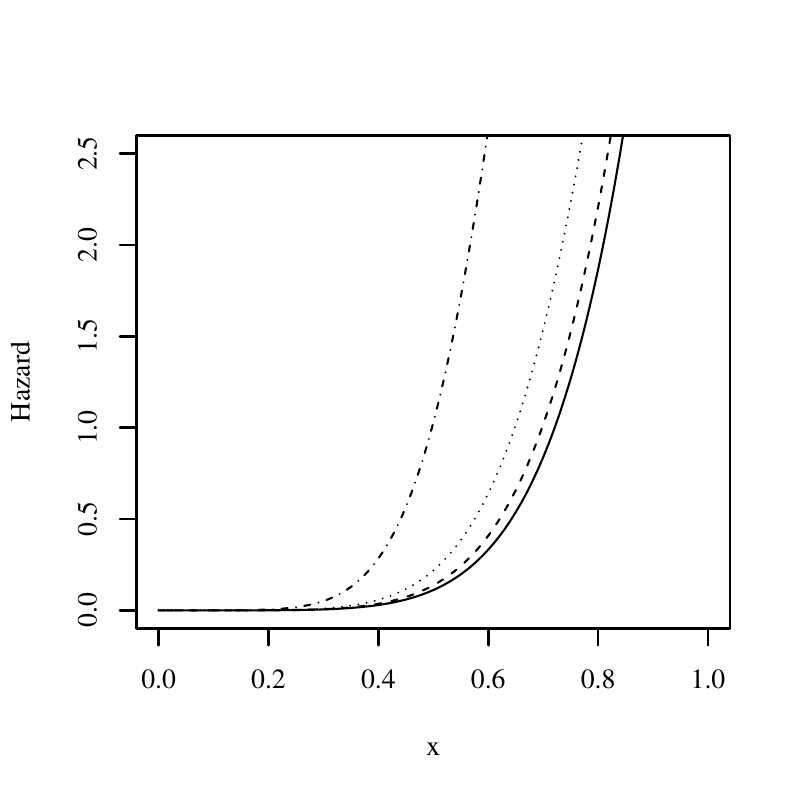}}&\hspace{-0.6cm}
\subfigure[\scriptsize{$\alpha= 0.1, \, \gamma = 6\,\,\mbox{and}\,\, \lambda=0.5$} \label{hz3}]{\includegraphics[width=0.37\textwidth]{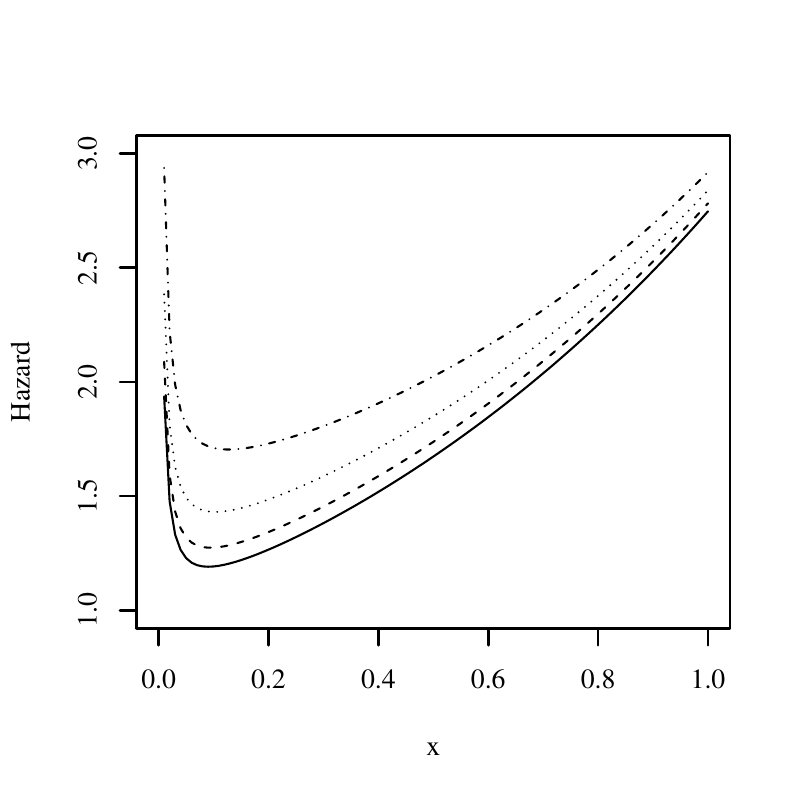}}\\
\end{tabular}
\vspace{-0.25cm}
\caption{Plots of the MWG hazard rate functions for $\theta = 0.01$ (solid line), $\theta = 0.2$ (dashed
line), $\theta = 0.5$ (dotted line) and $\theta = 0.9$ (dotdash line).}
\label{fig:hazardfigewps}
\end{figure}

The density of the $i$th order statistic $X_{i:m}$ in a random sample of size $m$ from the MWG distribution is given by (for $i = 1, \ldots, m$)
\begin{equation*}
f_{i:m}(x) = \frac{m!}{(i-1)! (m-i)!} \sum_{n=1}^\infty
\sum_{j=0}^{i-1} \omega_j\, p_n\,
\left[\frac{(1-\theta)\exp\left(-\alpha x^\gamma \mathrm{e}^{\lambda
x}\right)}{1 - \theta \exp\left(-\alpha x^\gamma \mathrm{e}^{\lambda
x}\right)}\right]^{m+j-i} g(x; n\alpha, \gamma, \lambda),
\end{equation*}
where $g(x; n\alpha, \gamma, \lambda)$ denotes the MW density function with parameters $n\alpha, \gamma$ and $\lambda$. From~\eqref{eosewps}, we obtain 
\begin{equation*}
\operatorname{E}\left(X_{i:m}^s\right) = \frac{m!}{(i-1)! (m-i)!}\sum_{n=1}^\infty \sum_{j=0}^{i-1} \omega_j\, p_n\,
\operatorname{E}\left\{X^s \left[\frac{(1-\theta)\exp\left(-\alpha
X^\gamma \mathrm{e}^{\lambda X}\right)}{1 - \theta \exp\left(-\alpha
X^\gamma \mathrm{e}^{\lambda X}\right)}\right]^{m+j-i}\right\}.
\end{equation*}

\subsection{Pareto Poisson distribution}
\vskip3mm
The Pareto Poisson (PP) distribution is defined by taking $H(x;\, \boldsymbol{\xi}) = \log(x/k)$ and $C(\theta) = \mathrm{e}^{\theta}-1$ in~\eqref{cdf}, which yields
\begin{equation*}
F(x; \theta,\alpha, k) = 1 - \frac{\exp\left[\theta \left(k/x\right)^{\alpha}\right]-1}{\mathrm{e}^{\theta}-1}, \quad x\geq k.
\end{equation*}

The pdf and hazard functions of the PP distribution are
\begin{equation*}
f(x; \theta,\alpha, k) = \frac{\theta\,\alpha\, k^{\alpha} \exp\left[\theta \left(k/x\right)^{\alpha}\right]}{(\mathrm{e}^{\theta}-1)\,x^{\alpha+1}}
\end{equation*}
and
\begin{equation*}
\tau(x; \theta,\alpha, k) = \frac{\theta\,\alpha \,k^{\alpha}\exp\left[\theta\left(k/x\right)^{\alpha}\right]}{x^{\alpha+1}\left\{\exp\left[\theta\left(k/x\right)^{\alpha}\right]-1\right\}}.
\end{equation*}

We obtain the Pareto distribution as a sub-model when $\theta \rightarrow 0$. The $r$th moment of the random variable $X$ following the PP distribution becomes
\begin{equation}\label{paretomoments}
E(X^r) = \frac{\alpha k^r}{(\mathrm{e}^{\theta}-1)}\sum_{n=1}^{\infty}\frac{\theta^n}{(n-1)!\,(n\alpha-r)},
\quad n\alpha > r.
\end{equation}

In particular, setting $r = 1$ in~\eqref{paretomoments}, the mean of $X$ reduces to
$$\mu =  \frac{\alpha k}{\mathrm{e}^{\theta}-1}\sum_{n=1}^{\infty}\frac{\theta^n}{(n-1)!\,(n\alpha-1)}, \quad n\alpha > 1.$$

\begin{figure}[!htbp]
\centering
\begin{tabular}{lll}
\hspace{-0.7cm}\subfigure[\scriptsize{$\alpha= 2, \,\gamma = 1.5\,\,\mbox{and}\,\, \lambda=0.5$} \label{ppd1}]{\includegraphics[width=0.37\textwidth]{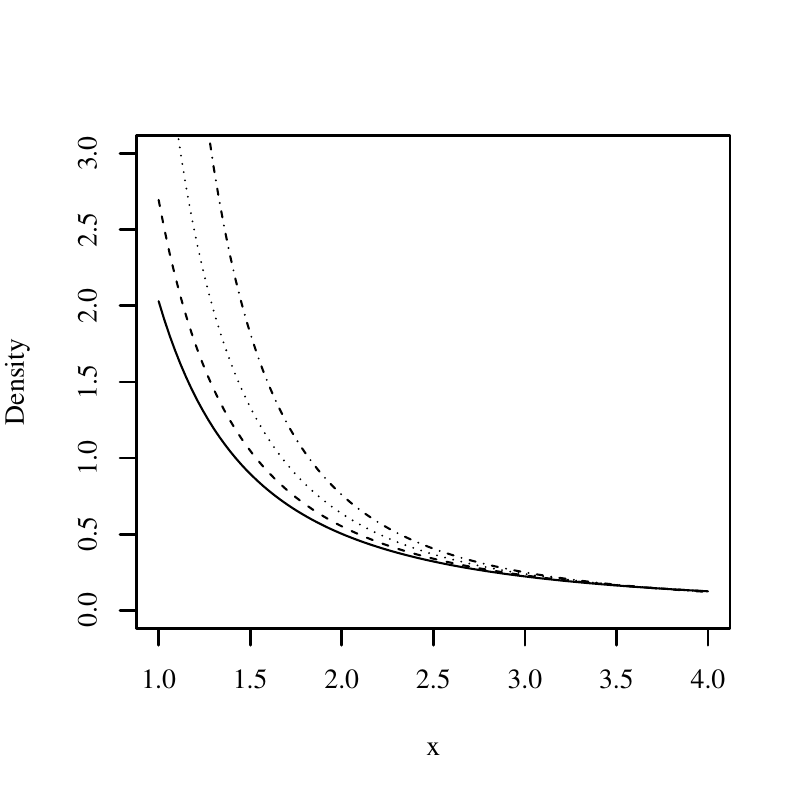}}&\hspace{-0.6cm}\subfigure[\scriptsize{$\alpha= 2, \, \gamma = 0.8\,\,\mbox{and}\,\, \lambda=0.01$} \label{ppd2}]{\includegraphics[width=0.37\textwidth]{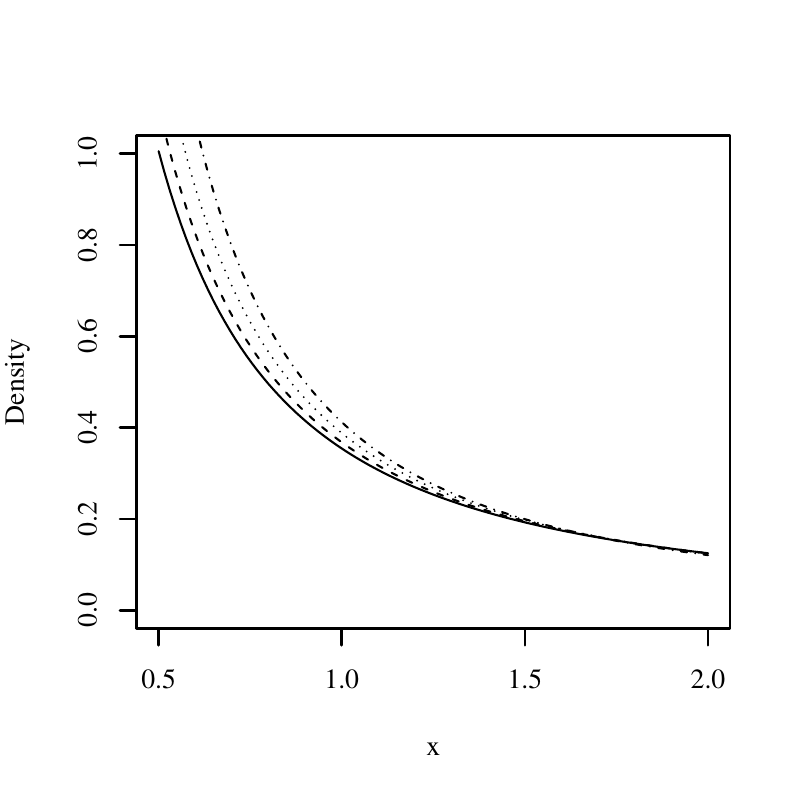}}&\hspace{-0.6cm}
\subfigure[\scriptsize{$\alpha= 0.1, \, \gamma = 6\,\,\mbox{and}\,\, \lambda=0.5$} \label{ppd3}]{\includegraphics[width=0.37\textwidth]{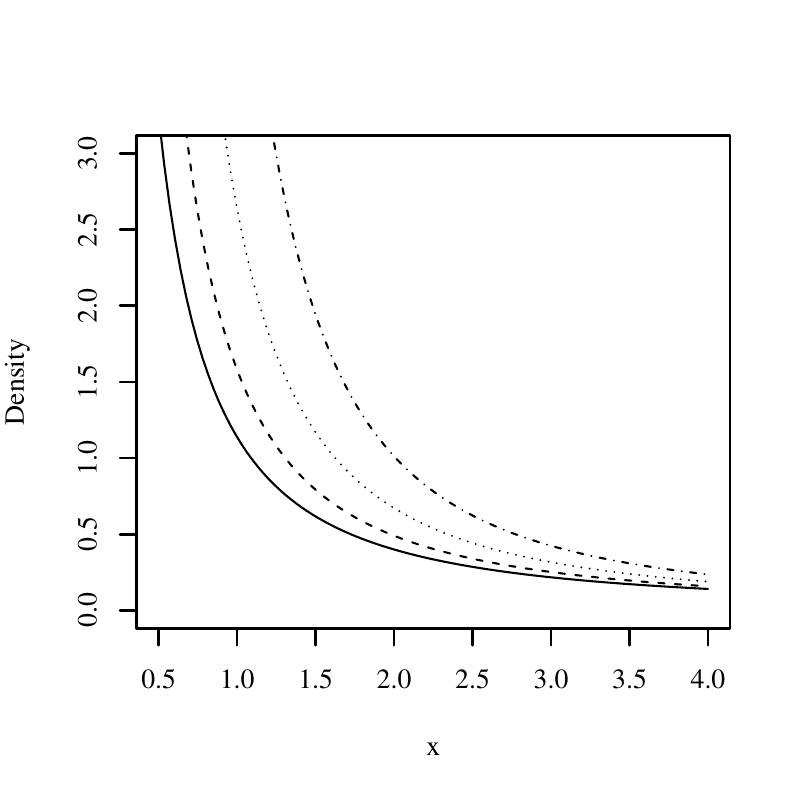}}\\
\end{tabular}
\vspace{-0.25cm}
\caption{Plots of the PP density functions for $\theta = 0.01$ (solid line), $\theta = 0.2$ (dashed line), $\theta = 0.5$ (dotted line) and $\theta = 0.9$ (dotdash line).}
\label{fig:ppdensities}
\end{figure}

\begin{figure}[!htbp]
\centering
\begin{tabular}{lll}
\hspace{-0.7cm}\subfigure[\scriptsize{$\alpha= k = 0.5$} \label{pph1}]{\includegraphics[width=.37\linewidth]{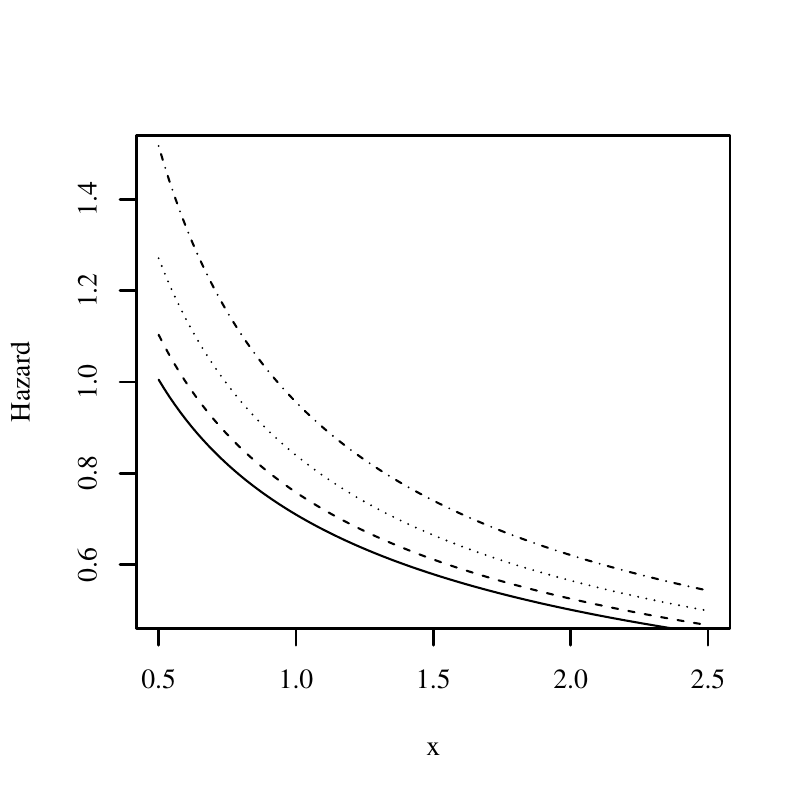}}&\hspace{-0.7cm}
\subfigure[\scriptsize{$\alpha= 2 \,\,\mbox{and}\,\, k = 1$} \label{pph2}]{\includegraphics[width=.37\linewidth]{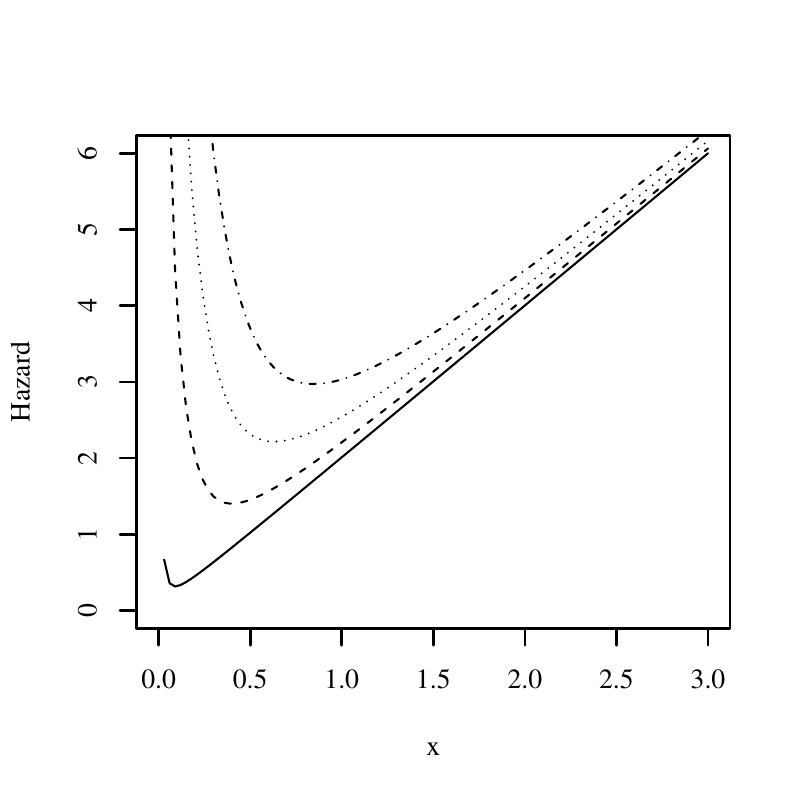}}&\hspace{-0.7cm}
\subfigure[\scriptsize{$\alpha= 7 \,\,\mbox{and}\,\, k=2$} \label{pph3}]{\includegraphics[width=.37\linewidth]{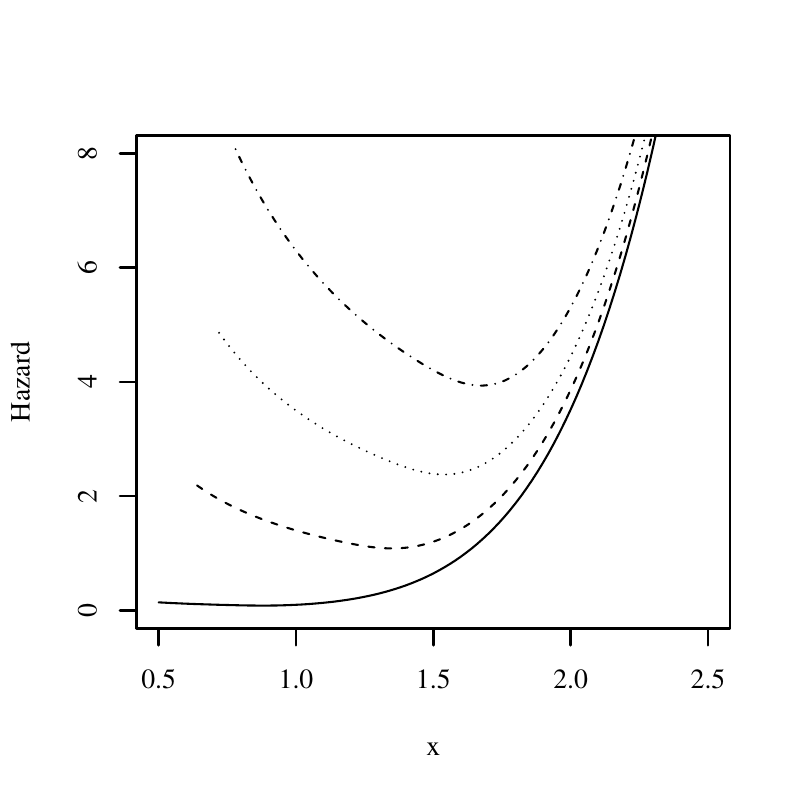}}\\
\end{tabular}
\vspace{-0.25cm}
\caption{Plots of the PP hazard functions for $\theta = 0.01$ (solid line), $\theta = 0.2$ (dashed line), $\theta = 0.5$ (dotted line) and $\theta = 0.9$ (dotdash line).}
\label{fig:pphazards}
\end{figure}

From equation (\ref{eosewps}), the $s$th moment of the $i$th order statistic, for $i = 1, \ldots, m,$ is given by
\begin{equation*}
\operatorname{E}\left(X_{i:m}^s\right) = \frac{m!}{(i-1)! (m-i)!}
\sum_{n=1}^\infty \sum_{j=0}^{i-1} \omega_j\, p_n\,
\operatorname{E}\left[X^s \left(\frac{\exp(\theta \left(k/X\right)^{\alpha})-1}{\mathrm{e}^{\theta}-1}\right)^{m+j-i}\right],
\end{equation*}
where $p_n$ denotes the Poisson probability function. Furthermore, after some algebra, the Shannon entropy for the PP distribution reduces to
\begin{equation*}
\mathbb{H}_{Sh}(f) = \log\left(\frac{\mathrm{e}^\theta-1}{\theta \alpha}\right) - \frac{\theta}{\mathrm{e}^\theta-1}\left(\mu_1 - \alpha \mu_2 + \mu_3\right),
\end{equation*}
where
\begin{align*}
\mu_1 &= \operatorname{E}\left[\exp\left\{\theta \left(\frac{k}{X}\right)^\alpha\right\}\log \left(\frac{1}{X}\right)\right] = \frac{1}{2(\mathrm{e}^\theta-1)}\left\{\frac{\mathrm{Chi}(2\theta)-\log(2\theta)+\mathrm{Shi}(2\theta)-\gamma}{\alpha} - (\mathrm{e}^{2\theta}-1)\log k\right\},\\
\mu_2 &= \operatorname{E}\left[\exp\left\{\theta \left(\frac{k}{X}\right)^\alpha\right\}\log \left(\frac{X}{k}\right)\right] = \frac{\mathrm{Chi}(2\theta)-\log(2\theta)+\mathrm{Shi}(2\theta)-\gamma}{2\alpha(\mathrm{e}^\theta-1)}
\intertext{and}
\mu_3 &= \operatorname{E}\left[\theta \exp\left\{\theta \left(\frac{k}{X}\right)^\alpha\right\}\left(\frac{k}{X}\right)^\alpha\right] = \frac{\alpha\,\theta\, k^{2\alpha}}{4(\mathrm{e}^{\theta}-1)}\left\{1- (2\theta+1)\mathrm{e}^{2\theta}\right\},
\end{align*}
where $$\operatorname{Chi}(z) = \gamma + \log z + \int_{0}^{z}\frac{\operatorname{cosh}(t)-1}{t}dt$$ is the hyperbolic cosine integral, $$\operatorname{Shi}(z) = \int_{0}^{z}\frac{\operatorname{sinh}(t)-1}{t}dt$$ is the hyperbolic sine integral and $\gamma \approx 0.577216$ is the Euler-Mascheroni constant.
\subsection{Chen logarithmic distribution}
\vskip3mm
The Chen logarithmic (CL) distribution is defined by the cdf (\ref{cdf}) with $H(x;\, \boldsymbol{\xi}) = \exp(x^\beta)-1$ and $C(\theta) = -\log(1-\theta)$, leading to
\begin{equation*}
F(x) = 1 - \frac{\log\left\{1-\theta \exp\left[-\alpha (\exp(x^\beta)-1)\right]\right\}}{\log(1-\theta)}   , \quad x > 0,
\end{equation*}
where $\theta \in (0,1)$. The associated pdf and hazard rate function (for $x>0$) are
\begin{equation*}
f(x) = \frac{\theta \alpha b x^{b-1} \exp\left\{x^b - \alpha\left[\exp(x^b)-1\right]\right\}}{\log(1-\theta)\left\{\theta \exp\left[- \alpha(\exp(x^b)-1)\right]-1\right\}}
\end{equation*}
and
\begin{equation*}
\tau(x) = \frac{\theta \alpha b x^{b-1}\exp\left[x^b - \alpha (\exp(x^b)-1)\right]}{\left\{\theta \exp\left[- \alpha (\exp(x^b)-1)\right]-1\right\}\log\left\{1-\theta \exp\left[- \alpha (\exp(x^b)-1)\right]\right\}},
\end{equation*}
respectively.

\begin{figure}[!htbp]
\centering
\begin{tabular}{lll}
\hspace{-0.7cm}\subfigure[\scriptsize{$\alpha= b = 1$} \label{cld1}]{\includegraphics[width=.37\linewidth]{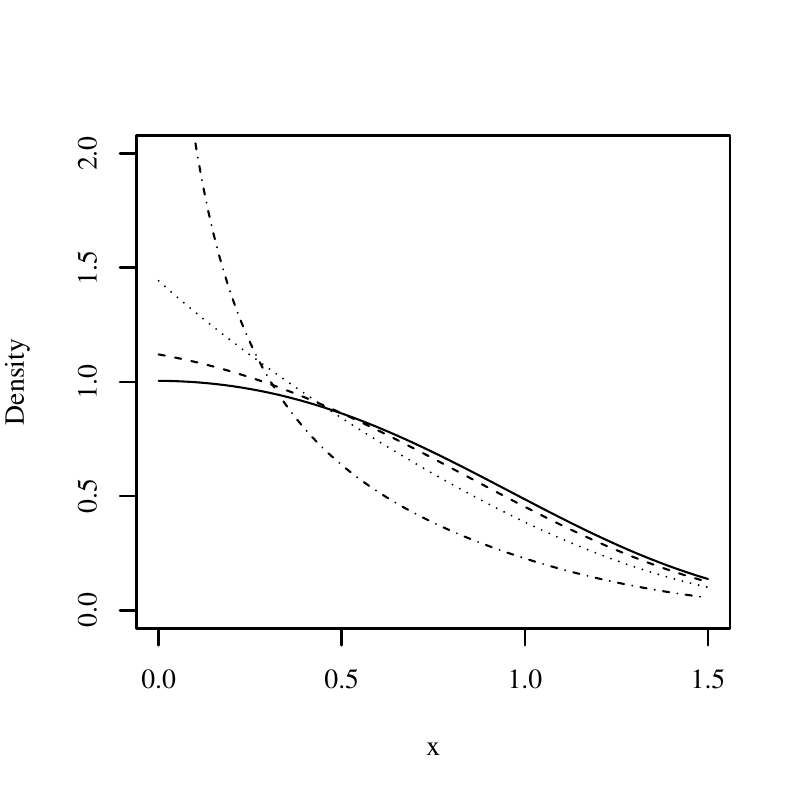}}&\hspace{-0.7cm}
\subfigure[\scriptsize{$\alpha= b = 1.5$} \label{cld2}]{\includegraphics[width=.37\linewidth]{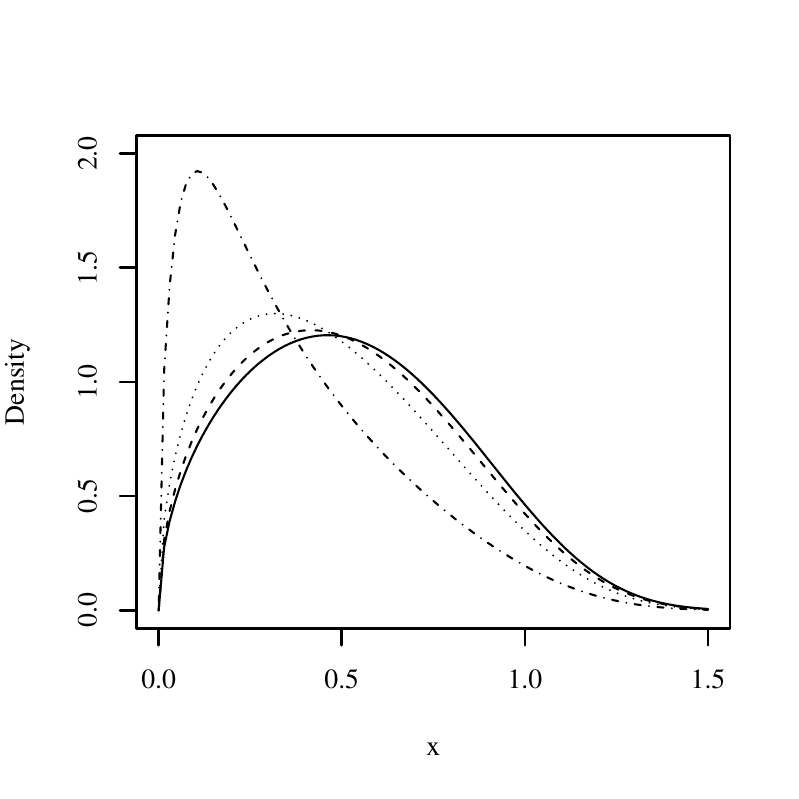}}&\hspace{-0.7cm}
\subfigure[\scriptsize{$\alpha= 2.5 \,\,\mbox{and}\,\, b=3$} \label{cld3}]{\includegraphics[width=.37\linewidth]{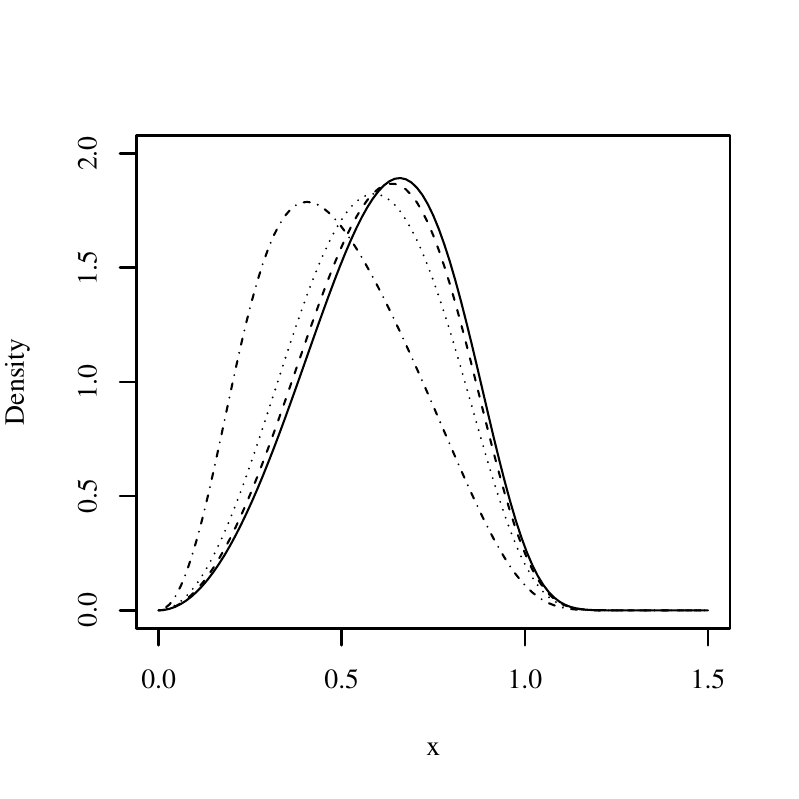}}\\
\end{tabular}
\vspace{-0.25cm}
\caption{Plots of the CL density functions for $\theta = 0.01$ (solid line), $\theta = 0.2$ (dashed line), $\theta = 0.5$ (dotted line) and $\theta = 0.9$ (dotdash line).}
\label{fig:cldensities}
\end{figure}

As expected by proposition~\ref{prop1}, we obtain the Chen distribution as a limiting special case when $\theta \rightarrow 0^+$.

\begin{figure}[!htbp]
\centering
\begin{tabular}{lll}
\hspace{-0.7cm}\subfigure[\scriptsize{$\alpha= 2 \,\, \mbox{and} \,\, b = 1$} \label{clh1}]{\includegraphics[width=.37\linewidth]{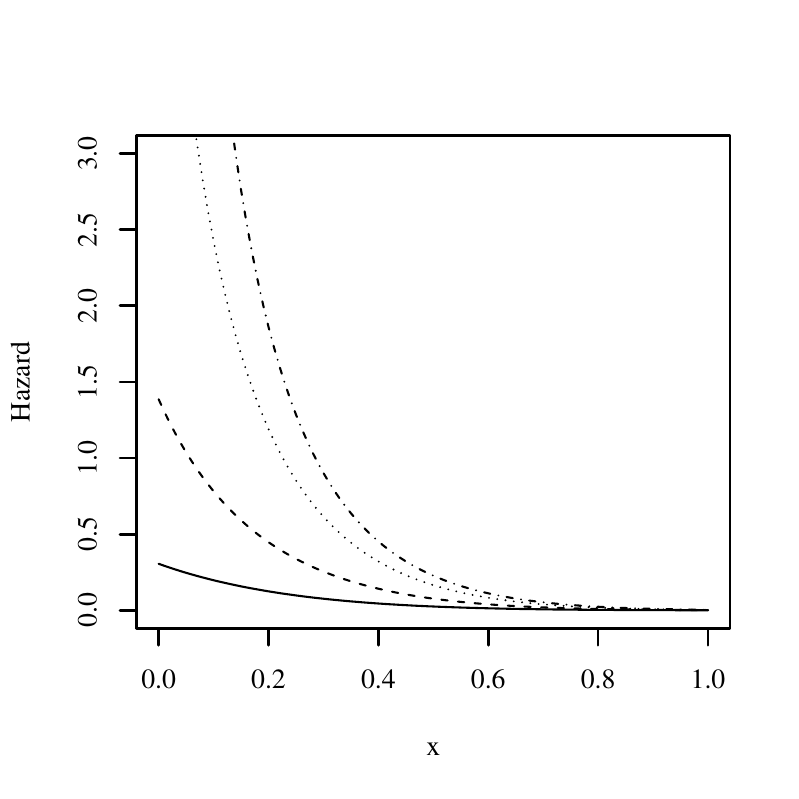}}&\hspace{-0.7cm}
\subfigure[\scriptsize{$\alpha= 3.5 \,\,\mbox{and} \,\, b= 7$} \label{clh2}]{\includegraphics[width=.37\linewidth]{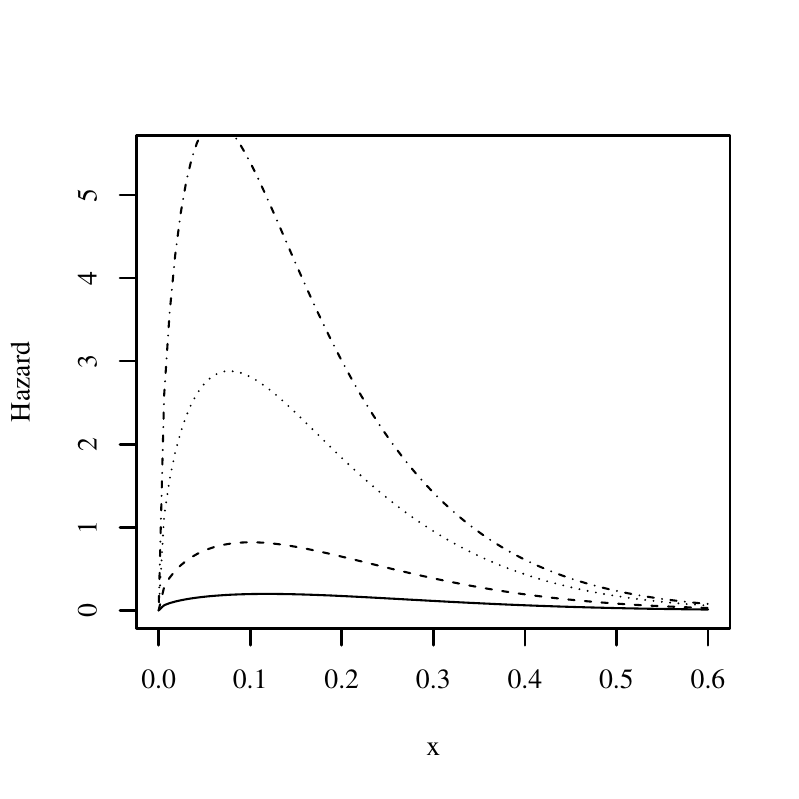}}&\hspace{-0.7cm}
\subfigure[\scriptsize{$\alpha= 3.5 \,\,\mbox{and}\,\, b=1.5$} \label{clh3}]{\includegraphics[width=.37\linewidth]{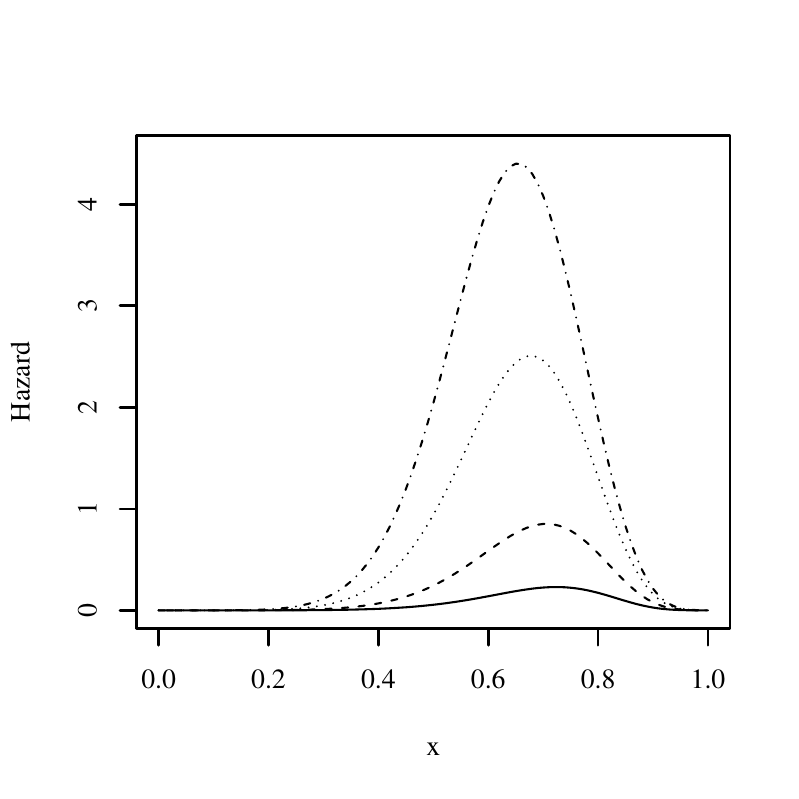}}\\
\end{tabular}
\vspace{-0.25cm}
\caption{Plots of the CL hazard rate functions for $\theta = 0.01$ (solid line), $\theta = 0.2$ (dashed line), $\theta = 0.5$ (dotted line) and $\theta = 0.9$ (dotdash line).}
\label{fig:clhazards}
\end{figure}

The density of the $i$th order statistic $X_{i:m}$ in a random sample of size $m$ from the CL distribution is given by (for $i = 1, \ldots, m$)
\begin{equation*}
f_{i:m}(x) = \frac{m!}{(i-1)! (m-i)!} \sum_{n=1}^\infty \sum_{j=0}^{i-1} \omega_j^*\, p_n\, g(x; n\alpha, b) \left\{\log\left[1-\theta\exp(\alpha-\alpha\operatorname{e}^{x^b})\right]\right\}^{m+j-1},
\end{equation*}
where $g(x; n\alpha, b)$ is the pdf of the Chen distribution with parameters $n\alpha$ and $b$ and $p_n$ denotes the logarithmic probability mass function and $$\omega_j^* = (-1)^j \binom{i-1}{j}\left[\frac{1}{\log(1-\theta)}\right]^{m+j-1}.$$ 
In the same way, the $s$th raw moment of $X_{i:m}$ is obtained directly from 
\begin{equation*}
\operatorname{E}\left(X_{i:m}^s\right) = \frac{m!}{(i-1)! (m-i)!}
\sum_{n=1}^\infty \sum_{j=0}^{i-1} \omega_j\, p_n\,
\operatorname{E}\left\{Z^s \exp\left[n\alpha(m+j-1)(1-\exp(Z^b))\right]\right\},
\end{equation*}
where $Z\sim \mbox{Chen}(n\alpha, b)$. 
\section{Application}

Fonseca and França (2007) studied the soil fertility influence and the characterization of the biologic fixation of $\mathrm{N}_2$ for the \emph{Dimorphandra wilsonii rizz growth}. For 128 plants, they made measures of the phosphorus concentration in the leaves. The data are listed in Table \ref{tab1aplic}. We fit the MWG, Gompertz Poisson (GP), PP,  Chen Poisson (CP) and CL models to these data. We also fit the three-parameter WG distribution introduced by Barreto-Souza \emph{et al}. (2010). The required numerical evaluations are implemented using the SAS (PROCNLMIXED) and R softwares.

\begin{table}[!htbp]
\centering
\footnotesize
\begin{tabular}{rrrrrrrrrrrrr}
\toprule
0.22& 0.17& 0.11& 0.10& 0.15& 0.06& 0.05& 0.07& 0.12& 0.09& 0.23& 0.25& 0.23\\
0.24& 0.20& 0.08& 0.11& 0.12& 0.10& 0.06& 0.20& 0.17& 0.20& 0.11&
0.16& 0.09\\
0.10& 0.12& 0.12& 0.10& 0.09& 0.17& 0.19& 0.21& 0.18& 0.26& 0.19&
0.17& 0.18\\
0.20& 0.24& 0.19& 0.21& 0.22& 0.17& 0.08& 0.08& 0.06& 0.09& 0.22&
0.23& 0.22\\
0.19& 0.27& 0.16& 0.28& 0.11& 0.10& 0.20& 0.12& 0.15& 0.08& 0.12&
0.09& 0.14\\
0.07& 0.09& 0.05& 0.06& 0.11& 0.16& 0.20& 0.25& 0.16& 0.13& 0.11&
0.11& 0.11\\
0.08& 0.22& 0.11& 0.13& 0.12& 0.15& 0.12& 0.11& 0.11& 0.15& 0.10&
0.15& 0.17\\
0.14& 0.12& 0.18& 0.14& 0.18& 0.13& 0.12& 0.14& 0.09& 0.10& 0.13&
0.09& 0.11\\
0.11& 0.14& 0.07& 0.07& 0.19& 0.17& 0.18& 0.16& 0.19& 0.15& 0.07&
0.09& 0.17\\
0.10& 0.08& 0.15& 0.21& 0.16& 0.08& 0.10& 0.06& 0.08& 0.12& 0.13\\
\bottomrule
\end{tabular}
\caption{Phosphorus concentration in leaves data set.}\label{tab1aplic}
\end{table}

\begin{table}[!htbp]
\centering
\footnotesize
\begin{tabular}{ccccccc}
\toprule
Min.  & $Q_{1}$   &  $Q_{2}$&    Mean &  $Q_{3}$  &Max.   &Var.\\
0.0500 &   0.1000 &  0.1300 &   0.1408&  0.1800  &0.2800 &0.0030\\
\bottomrule
\end{tabular}
\caption{Descriptive statistics. }\label{tab2aplic}
\end{table}

Tables \ref{tab2aplic} and \ref{tab3aplic} display some descriptive statistics and the MLEs (with corresponding standard errors in parentheses) of the model parameters. Since the values of the Akaike information criterion (AIC), Bayesian information criterion (BIC) and consistent Akaike information criterion (CAIC) are smaller for the CL distribution compared with those values of the other models, this new distribution seems to be a very competitive model for these data.

\begin{table}[!htbp]
\footnotesize
\centering
\begin{tabular}{lcccccccc}
\toprule
\multicolumn{4}{r}{Estimates}&&\multicolumn{3}{r}{Statistic}\\
\cmidrule{2-5} \cmidrule{7-9}
Model   &$\theta$    &$\alpha$   &$\gamma$ &$\lambda$    & &AIC    &BIC  & AICC\\
\midrule
MWG     &0.7200      &409.07 &3.6545
&$-$0.5727  & &$-$385.6
  &$-$374.2
&$-$385.3
   \\
        &(0.2418)    &(1174.76)  &(0.821)  &(6.6673)     & &       &     &    \\
WG      &0.9995      &2.4471     &4.2041   &$-$ &
&$-$378.5&$-$370.0 & $-$378.3
    \\
       &(0.0017)  &(8.7059)   &(0.3022) &    $-$         & &        &     &
       \\ \cmidrule{2-4}
                &$\theta$             &$\alpha$           & $\beta$ &     & &       &     &
                \\ \cmidrule{2-4}
GP      &2.9478      &0.3169     &19.7047  &     &
& $-$368.7&
        $-$360.2& $-$368.5   \\
                &(1.2627)    &(0.1473)   &(1.6135) &            & &       &     &
                \\ \cmidrule{2-4}
         &$\theta$      & $\alpha$         & $k$      &       &     & &
         &\\ \cmidrule{2-4}
PP &80.0903       &0.0131           &0.0500    &
&&$-$271.4&$-$265.7& $-$271.3\\
         &(69.7770)     &(0.0115)         &      &       &&&&\\
\cmidrule{2-4}
 &$\theta$      & $\alpha$         & $b$      &       &     & &         &\\
\cmidrule{2-4}
CP&15.4386&14.7817& 2.9212& && $-$383.7&
$-$375.2&$-$383.5\\
        &(22.8318)  &(28.1576)&(0.2634)& &&  & &\\
CL &0.9999& 52232&7.5882&  && $-$395.8&$-$387.2
 &$-$395.6\\
 &(0.0001)&(0.0000)&(0.2039)& && & &\\
 \bottomrule
\end{tabular}
\caption{MLEs of the model parameters, the corresponding SEs (given in parentheses) and the statistics AIC, BIC and AICC.}\label{tab3aplic}
\end{table}

Plots of the estimated pdf and cdf of the MWG, WG, GP, PP, CP and CL models fitted to these data are displayed in Figure \ref{fig:cdfplot}. They indicate that the CL distribution is superior to the other distributions in terms of model fitting.

\begin{figure}[!htbp]
\centering
\begin{tabular}{ll}
\hspace{-0.7cm}\subfigure[]{\epsfig{file=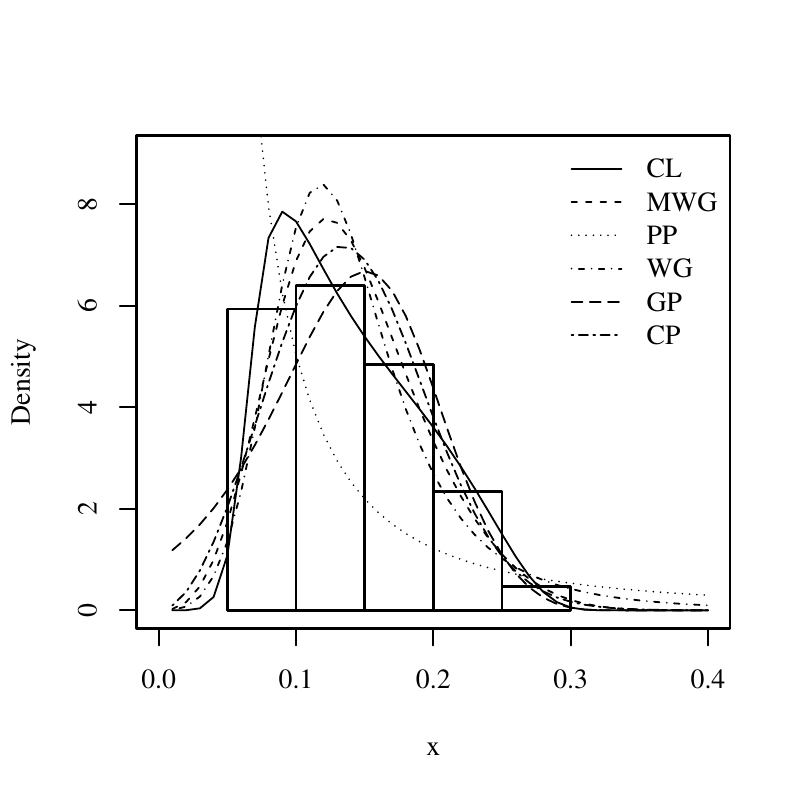,height=0.54\linewidth}}&\hspace{-0.6cm}
\subfigure[]{\epsfig{file=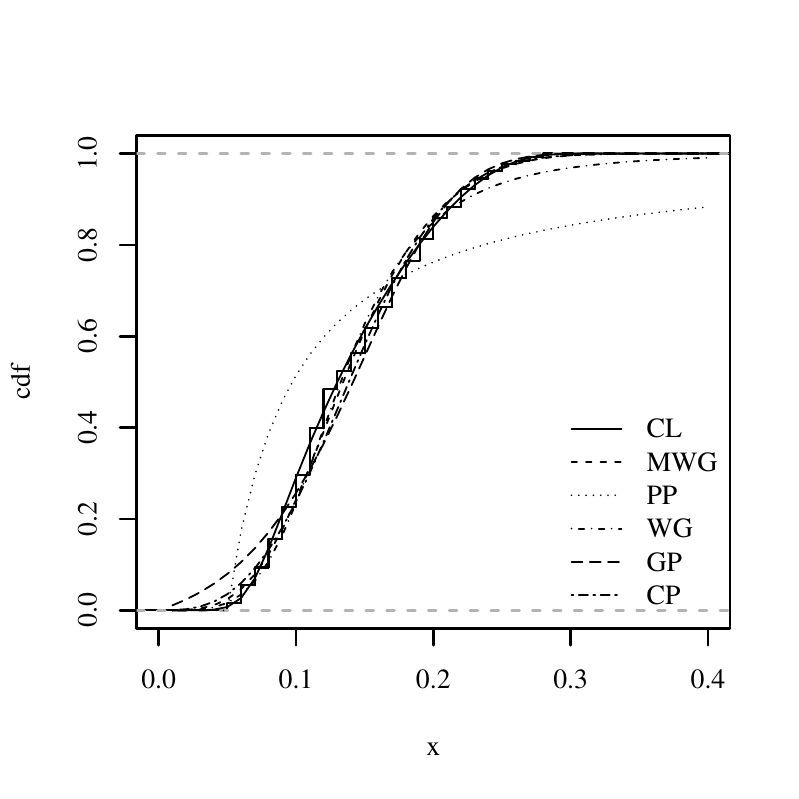,height=0.54\linewidth}}\\
\end{tabular}
\caption{Estimated (a)  pdf and (b) cdf for the CL, MWG, PP, WG, GP and CP models to the percentage of Phosphorus concentration in leaves data.}
\label{fig:cdfplot}
\end{figure}


Table \ref{tab4aplic} lists the values of the Kolmogorov-Smirnov (K-S) statistic and the values of $-2\ell(\widehat{\Theta})$. From these figures, we conclude that the CL distribution provides a better fit to these data than the MWG, WG, GP, PP and CP models.

\begin{table}[!htb]
\footnotesize
\centering
\begin{tabular}{lcccccccccccc}
\toprule Model         &&&& K--S   &&&& $-2\ell(\widehat{\Theta})$\\
\midrule
MWG  &&&&0.0943        &&&&$-$393.6             \\
WG           &&&&0.0873        &&&&$-$384.5            \\
Gompertz Poisson            &&&&0.1201         &&&&$-$374.7            \\
Pareto Poisson              &&&&0.3017         &&&&$-$374.7                \\
Chen Poisson                &&&&0.1159         &&&&$-$389.7                     \\
Chen Logarithmic            &&&&0.0678        &&&&$-$401.8            \\
 \bottomrule
\end{tabular}
\caption{K-S statistics and $-2\ell(\widehat{\Theta})$ for the exceedances of phosphorus concentration in leaves data set.}\label{tab4aplic}
\end{table}

\section{Concluding remarks}

We define a new lifetime class of distributions, called the extended Weibull power series (EWPS), which generalizes the Weibull power series class of distributions proposed by Morais and Barreto-Souza (2011), which in turn extends the exponential power series class of distributions (Chahkandi and Ganjali, 2009). We provide a mathematical treatment of the new
distribution including expansions for the density function, moments, generating function and incomplete moments. Further, explicit expressions for the order statistics and Shannon entropy are derived. The EWPS density function can be expressed as a mixture of EW density functions. This property is important to obtain several other results. Our formulas related with the EWPS model are manageable, and with the use of modern computer resources with analytic and numerical capabilities, they may turn into adequate tools comprising the arsenal of applied statisticians. The estimation of the model parameters is approached by the method of maximum likelihood using the EM algorithm. The observed information matrix is derived. Further, maximum entropy identification for the EWPS distributions was discussed and some special models are studied in some detail. Finally, we fit the EWPS model to a real data set to show the usefulness of the proposed class. We hope that this generalization may attract wider applications in the literature of the fatigue life distributions.

\section*{Acknowledgements}
\addcontentsline{toc}{chapter}{Acknowledgements}  We also gratefully acknowledge financial support from CAPES and CNPq.
\appendix

\section{}
\addcontentsline{toc}{chapter}{Appendix A.}\label{apA}
The elements of the $p \times p$ information matrix $J_n(\Theta)$ are
\begin{align*}
J_{\theta \theta} &= - \frac{n}{\theta^2} - n\left[\frac{C''(\theta)}{C(\theta)} - \left(\frac{C'(\theta)}{C(\theta)}\right)^2\right] + \theta \sum_{i=1}^n \left(\frac{z_{2i}}{z_{1i}}\right)^2 H(x_i;\, \xi) \mathrm{e}^{-2\alpha H(x_i;\, \xi)}\\
&- \theta \sum_{i=1}^n \frac{z_{3i}}{z_{1i}} H(x_i;\, \boldsymbol{\xi}) \mathrm{e}^{-2\alpha H(x_i;\, \boldsymbol{\xi})}\\
J_{\alpha \alpha} &= - \frac{n}{\alpha^2} + \theta \sum_{i=1}^n \frac{z_{2i}}{z_{1i}} H^2(x_i; \boldsymbol{\xi}) \mathrm{e}^{-\alpha H(x_i;\, \xi)} + \theta^2 \sum_{i=1}^n \frac{(z_{3i}-z_{2i}^2)}{z_{1i}}H^2(x_i; \xi) \mathrm{e}^{-2\alpha H(x_i;\, \xi)}\\
J_{\alpha \theta} &= \theta \sum_{i=1}^n \left[\left(\frac{z_{2i}}{z_{1i}}\right)^2 - \frac{z_{3i}}{z_{1i}}\right] H^2(x_i; \xi) \mathrm{e}^{-2\alpha H(x_i;\, \xi)} - \sum_{i=1}^n \frac{z_{2i}}{z_{1i}} H^2(x_i; \xi) \mathrm{e}^{-\alpha H(x_i;\, \xi)} \\
J_{\alpha \boldsymbol{\xi}_k} &= - \sum_{i=1}^n \frac{\partial H(x_i;\, \boldsymbol{\xi})}{\partial \xi_k} - \theta \sum_{i=1}^n \frac{z_{2i}}{z_{1i}}\frac{\partial H(x_i;\, \xi)}{\partial \xi_k} \mathrm{e}^{-\alpha H(x_i;\, \boldsymbol{\xi})} \left[1 - \alpha H(x_i;\, \boldsymbol{\xi})\right] \\
&+ \alpha \theta^2 \sum_{i=1}^n \left[\frac{z_{3i}}{z_{1i}}-\left(\frac{z_{2i}}{z_{1i}}\right)^2 \right] \frac{\partial H(x_i;\, \xi)}{\partial \boldsymbol{\xi}_k}H(x_i;\, \xi)\mathrm{e}^{-2\alpha H(x_i;\, \boldsymbol{\xi})}\\
J_{\theta \boldsymbol{\xi}_k} &= \theta \alpha  \sum_{i=1}^n \left[\left(\frac{z_{2i}}{z_{1i}}\right)^2 - \frac{z_{3i}}{z_{1i}}\right] \frac{\partial H(x_i;\, \boldsymbol{\xi})}{\partial \boldsymbol{\xi}_k} \mathrm{e}^{-2\alpha H(x_i;\, \xi)} - \alpha \sum_{i=1}^n \frac{z_{2i}}{z_{1i}} \frac{\partial H(x_i;\, \boldsymbol{\xi})}{\partial \boldsymbol{\xi}_k} \mathrm{e}^{-\alpha H(x_i;\, \boldsymbol{\xi})}\\
J_{\xi_k \boldsymbol{\xi}_l} &= -\alpha \sum_{i=1}^n \frac{\partial^2 H(x_i;\, \boldsymbol{\xi})}{\partial \boldsymbol{\xi}_k \partial \boldsymbol{\xi}_l} - \sum_{i=1}^n  \frac{1}{H(x_i;\, \boldsymbol{\xi})^2}\frac{\partial H(x_i;\, \boldsymbol{\xi})}{\partial \boldsymbol{\xi}_k}\frac{\partial H(x_i;\, \boldsymbol{\xi})}{\partial \boldsymbol{\xi}_l} + \sum_{i=1}^n  \frac{1}{H(x_i;\, \boldsymbol{\xi})}\frac{\partial^2 H(x_i;\, \boldsymbol{\xi})}{\partial \boldsymbol{\xi}_k \partial \boldsymbol{\xi}_l}\\
&+ (\alpha \theta)^2 \sum_{i=1}^n \left[\left(\frac{z_{2i}}{z_{1i}}\right)^2 + \frac{z_{3i}}{z_{1i}}\right] \frac{\partial H(x_i;\, \boldsymbol{\xi})}{\partial \xi_k} \frac{\partial H(x_i;\, \boldsymbol{\xi})}{\partial \boldsymbol{\xi}_l} \mathrm{e}^{-2\alpha H(x_i;\, \boldsymbol{\xi})} \\
&- \alpha \theta \sum_{i=1}^n \frac{z_{2i}}{z_{1i}} \frac{\partial^2
H(x_i;\, \xi)}{\partial \boldsymbol{\xi}_k \partial
\boldsymbol{\xi}_l} \mathrm{e}^{-\alpha H(x_i;\, \boldsymbol{\xi})} + \alpha^2
\theta \sum_{i=1}^n \frac{z_{2i}}{z_{1i}} \frac{\partial H(x_i;
\xi)}{\partial \boldsymbol{\xi}_k} \frac{\partial H(x_i;
\boldsymbol{\xi})}{\partial \boldsymbol{\xi}_l}\mathrm{e}^{-\alpha H(x_i;
\boldsymbol{\xi})}
\end{align*}
where $z_{1i} = C'(\theta e^{-\alpha H(x_i;\, \boldsymbol{\xi})}),
z_{2i} = C''(\theta \mathrm{e}^{-\alpha H(x_i;\, \boldsymbol{\xi})})$ and
$z_{3i} = C'''(\theta \mathrm{e}^{-\alpha H(x_i;\, \boldsymbol{\xi})})$, for $i
= 1, \ldots, n$.



\section*{References}
\begin{thebibliography}{99}
\addcontentsline{toc}{chapter}{References}

\bibitem[Adamidis and Loukas (2005)]{adamidisloukas2005} Adamidis K., Dimitrakopoulou, T., Loukas, S. (2005). On a generalization of the exponential-geometric distribution. \emph{Statistics \& Probability Letters}, {\bf 73}, 259--269.

\bibitem[Adamidis and Loukas (1998)]{adamidisloukas1998} Adamidis K., Loukas, S. (1998). A lifetime distribution with decreasing failure rate. \emph{Statistics \& Probability Letters}, {\bf 39}, 35--42.

\bibitem[Barreto-Souza and Cribari-Neto (2009)]{barretosouzacribarineto2009} Barreto-Souza, W., Cribari-Neto, F. (2009). A generalization of the exponential-Poisson distribution. \emph{Statistics \& Probability Letters}, {\bf 79}, 2493--2500.

\bibitem[Barreto-Souza and Morais (2009)]{barretosouzaetal2009} Barreto-Souza, W., Morais, A.L., Cordeiro, G.M. (2010). The Weibull-geometric distribution. \emph{Journal of Statistical Computation and Simulation}, {\bf 81}, 645--657.

\bibitem[Bebbington et al. (2007)]{bebbi2007} Bebbington, M., Lai, C. D. and Zitikis, R. (2007). A flexible
Weibull extension. \emph{ Reliability Engineering and System
Safety}, {\bf 92}, 719--726.

\bibitem[Carrasco et al. (2008)]{carrascoetal2008} Carrasco J.M.F., Ortega, E.M.M., Cordeiro, G.M. (2008). A generalized modified Weibull distribution for lifetime modeling. \emph{Computational Statistics \& Data Analysis}, {\bf 53}, 450--462.

\bibitem[Chakandi and Ganjali (2009)]{chahkandiganjali2009} Chahkandi, M., Ganjali, M. (2009). On some lifetime distributions with decreasing failure rate. \emph{Computational Statistics \& Data Analysis}, {\bf 53},  4433--4440.

\bibitem[Chen (2000)]{chen2000} Chen, Z. (2000). A new two-parameter lifetime distribution with
bathtub shape or increasing failure rate function. \emph{Statistics
and Probability Letters}, {\bf 49}, 155--161.

\bibitem[Dempster et al. (1977)]{dempsteretal1977} Dempster, A.P., Laird, N.M., Rubim, D.B. (1977). Maximum likelihood from incomplete data via the EM algorithm (with discussion). \emph{Journal of the Royal
Statistical Society. Series B}, {\bf 39}, 1--38.

\bibitem[Doornik (2007)]{doornik2007} Doornik, J. (2007) Ox 5: object-oriented matrix programming language, 5th ed. Timberlake Consultants, London.

\bibitem[Fonseca (2007)]{fonseca2007} Fonseca, M.B. and França, M.G.C. (2007). A influência da fertilidade
do solo e caracterização da fixação biológica de $\mathrm{N}_2$ para
o crescimento de \emph{Dimorphandra wilsonii rizz. Master's thesis,
Universidade Federal de Minas Gerais}.

\bibitem[Ghitany et al. (2011)]{ghitanyetal2011} Ghitany, M.E., Al-Jarallah, R.A., Balakrishnan, N. (2011): On the existence and uniqueness of the MLEs of the parameters of a general class of exponentiated distributions. \emph{Statistics}, \texttt{DOI:10.1080/02331888.2011.614950}.

\bibitem[Gompertz (1825)]{gomp1825} Gompertz, B. (1825). On the nature of the function expressive of the
law of human mortality and on the new model of determining the value
of life contingencies. Philosophical Trans. \emph{Royal Society of
London}, {\bf 115}, 513--585.

\bibitem[Gupta and Kundu (1999)]{guptakundu1999} Gupta, R.D. and Kundu, D. (1999). Generalized exponential distributions. \emph{Austral. NZ J. Statist.}, {\bf 41},  173--188.

\bibitem[Gupta and Kundu (2001a)]{guptakundu2001a} Gupta, R.D. and Kundu, D. (2001). Exponentiated exponential distribution: An alternative to gamma and Weibull
distributions. \emph{Biometrical Journal}, {\bf 43},  117--130.

\bibitem[Gupta and Kundu (2001b)]{guptakundu2001b} Gupta, R.D. and Kundu, D. (2001). Generalized exponential distributions: Different methods of estimations. \emph{Journal of Statistical Computation and Simulation}, {\bf 69},  315--338

\bibitem[Gupta and Kundu (2007)]{guptakundu2007} Gupta, R.D. and Kundu, D. (2007). Generalized exponential distributions: Existing results and some recent developments. \emph{Journal of Statistical Planning and  Inference}, {\bf 137},  3525--3536.

\bibitem[Gupta et al. (1998)]{guptaetal1998}Gupta, R.C., Gupta, R.D. and Gupta, P.L. 1998. Modeling failure time data by Lehman alternatives. \emph{Communications in Statistics:
Theory and Methods}, {\bf 27},  887--904.

\bibitem[Gurvich et al. (1997)]{gurvichetal1997} Gurvich, M., DiBenedetto, A., Ranade, S. (1997). A new statistical distribution for cha\-rac\-te\-ri\-zing the random strength of brittle materials. \emph{Journal of Materials Science}, {\bf 32},  2559--2564.

\bibitem[Jaynes (1957)]{jaynes1957} Jaynes, E.T. (1957). Information theory and statistical mechanics. \emph{Physical Reviews}, {\bf 106}, 620--630.

\bibitem[Johnson and Kotz (1994)]{JohnKotz1994} Johnson, N.L, Kotz, S. and Balakrishnan, N. (1994). Continuous
Univariate Distributions volume 1.  John Wiley \& Sons, New York.

\bibitem[Kapur (1989)]{kapur1989} Kapur, J.N. (1989). Maximum Entropy Models in Science and Engineering. John Wiley \& Sons, New York.

\bibitem[Kies (1958)]{kies1958} Kies, J. A. (1958). \emph{The strength of glass}. Washington D.C.
Naval Research Lab, 5093.

\bibitem[Kundu and Raqab (2005)]{kunduraqab2005} Kundu, D. and Raqab, M.Z. (2005). Generalized Rayleigh distribution: Different methods of estimation. \emph{Computational Statistics \& Data Analysis}, {\bf 49},  187--200.

\bibitem[Kus (2007)]{kus2007} Kus, C. (2007). A new lifetime distribution. \emph{Computational Statistics \& Data Analysis}, {\bf 51},  4497--4509.

\bibitem[Lai et al. (2003)]{lai2003} Lai, C.D., Xie, M. and D. N. P. Murthy. (2003). A modified weibull
distribution. \emph{Transactions on Reliability}, {\bf 52}, 33--37.

\bibitem[Lu and Shi (2011)]{lushi2011} Lu, W., Shi, D. (2011). A new compounding life distribution: the Weibull-Poisson distribution. \emph{Journal of Applied Statistics}, \texttt{DOI:10.1080/02664763.2011.575126}.
\bibitem[Morais and Barreto-Souza (2011)]{moraisbarretosouza2011} Morais, A.L., Barreto-Souza, W. (2011). A Compound Class of Weibull and Power Series Distributions. \emph{Computational Statistics \& Data Analysis},
{\bf 55},  1410--1425.

\bibitem[Nadarajah and Kotz (2005)]{nadkot2005} Nadarajah, S. and Kotz, S. (2005). On some recent modifications of
Weibull distribution. \emph{IEEE Trans. Reliability}, {\bf 54},
561--562.

\bibitem[Nikulin and Haghighi (2006)]{niku2006} Nikulin, M. and Haghighi, F. (2006). A Chi-squared test for the
generalized power Weibull family for the head-and-neck cancer
censored data. \emph{Journal of Mathematical Sciences}, {\bf 133},
1333--1341.

\bibitem[Pham (2002)]{pham2002} Pham, H. (2002). A vtub-shaped hazard rate function with
applications to sys- tem safety. \emph{International Journal of
Reliability and Applications}, {\bf 3}, 1--16.

\bibitem[Phani (1987)]{phani1987} Phani, K.K. (1987). A new modifiedWeibull distribution function.
\emph{Communications of the American Ceramic Society}, {\bf 70},
182--184.

\bibitem[Rayleigh (1880)]{ray1880} Rayleigh, J. W. S. (1880). On the resultant of a
large number of vibrations of the same pitch and of arbitrary phase.
\emph{Philosophical Magazine}, {\bf 10}, 73--78.

\bibitem[Shannon (1948)]{shannon1948} Shannon, C.E., (1948). A mathematical theory of communication. \emph{Bell System Technical Journal}, 27, 379--432.

\bibitem[Shore and Johnson (1980)]{shorejohnson1980} Shore, J.E., Johnson, R.W. (1980). Axiomatic derivation of the principle of maximum entropy and the principle of minimum cross-entropy. \emph{IEEE Transactions
on Information Theory}, {\bf 26}, 26--37.

\bibitem[Smith and Bain (1975)]{smith1975} Smith, R.M. and Bain, L.J. (1975). An exponential power life-testing
distribution. \emph{Communications Statistics}, {\bf 4}, 469--481.

\bibitem[Soofi (2000)]{soofi2000} Soofi, E.S. (2000). Principal information theoretic approaches. \emph{Journal of the American Statistical Association}, {\bf 95}, 1349--1353.

\bibitem[Surles and Padgett (2001)]{surlespadgett2001} Surles, J.G. and Padgett,  W.J. (2001). Inference for reliability and stress-strength for a scaled Burr type X distribution.
\emph{Lifetime Data Analysis.}, {\bf 7},  187--200.

\bibitem[Silva et al. (2010)]{silvaetal2010} Silva, R.B., Barreto-Souza, W., Cordeiro, G.M. (2010). A new distribution with decreasing, increasing and upside-down bathtub failure rate. \emph{Computational Statistics \& Data Analysis}, {\bf 54},  935--934.

\bibitem[Tamasbi and Rezaei (2008)]{tahmasbirezaei2008} Tahmasbi, R., Rezaei, S. (2008). A two-parameter lifetime distribution with decreasing failure rate. \emph{Computational Statistics \& Data Analysis}, {\bf 52},  3889--3901.

\bibitem[Xie and Lai (1995)]{xie1995} Xie, M. and Lai, C.D. (1995). Reliability analysis using additive
Weibull model with bathtub-shaped failure rate function.
\emph{Reliability Engineering and System Safety}, {\bf 52}, 87--93.

\bibitem[Xie et al. (2002)]{xie2002} Xie, M., Tang, Y. and Goh, T.N. (2002). A modified Weibull extension
with bathtub-shaped failure rate function. \emph{Reliability
Engineering \& System Safety}, {\bf 76}, 279--285.

\bibitem[Zografos and Balakrishnan (2009)]{zografosbalakrishnan2009} Zografos, K., Balakrishnan, N. (2009). On families of beta-and generalized gamma-generated distributions and associated inference. \emph{Statistical Methodology}, 6, 344--362.

\bibitem[White (1969)]{white1969} White, J. S. (1969). The moments of log-Weibull order statistics.
\emph{Technometrics}, {\bf 11}, 373--386.

\end{thebibliography}
\end{document}